\documentclass[conference]{IEEEtran}
\IEEEoverridecommandlockouts
\usepackage{amsmath,amssymb,amsfonts}
\usepackage{algorithmic}
\usepackage{textcomp}
\usepackage{microtype}
\usepackage{graphicx}
\usepackage{subfigure}
\usepackage{balance}
\usepackage{booktabs} 
\usepackage{amsthm}
\usepackage{array}
\usepackage{graphicx}
\usepackage{clrscode}
\usepackage{subfigure}
\usepackage{multirow}
\usepackage{multicol}
\usepackage{float}
\usepackage{color}
\usepackage{xcolor}
\usepackage{amsopn}
\usepackage{mathrsfs}
\usepackage{mathtools}
\usepackage{amsmath}
\usepackage{booktabs}
\usepackage{arydshln}
\usepackage{hyperref}
\usepackage{blkarray}
\usepackage{enumerate}
\usepackage{courier}
\usepackage{mathrsfs}
\usepackage{rotating}
\usepackage{bm}
\usepackage{subfigure}
\usepackage{array}
\usepackage{ragged2e}
\usepackage{hyperref}
\usepackage{amsmath}
\usepackage{threeparttable}

\def\BibTeX{{\rm B\kern-.05em{\sc i\kern-.025em b}\kern-.08em
    T\kern-.1667em\lower.7ex\hbox{E}\kern-.125emX}}

\usepackage[ruled,linesnumbered]{algorithm2e}

\SetKwComment{comment}{ $triangleright$ \ }{}
\theoremstyle{definition}
\newtheorem{definition}{Definition}
\theoremstyle{theorem}
\newtheorem{mytheory}{Theorem}

\theoremstyle{proof}

\theoremstyle{remark}
\newtheorem*{remark}{Remark}

\newtheorem{lemma}{Lemma}[section]
\AtBeginDocument{%
	\providecommand\BibTeX{{%
			\normalfont B\kern-0.5em{\scshape i\kern-0.25em b}\kern-0.8em\TeX}}}

\begin{document}
\title{Sequential Recommendation with User Causal Behavior Discovery}
\author{\IEEEauthorblockN{Zhenlei Wang$^{1,2,}$, Xu Chen$^{1,2,^*}$, Rui Zhou$^{1,2}$, Quanyu Dai$^{3}$, Zhenhua Dong$^{3}$, Ji-Rong Wen$^{1,2}$}\thanks{$^*$ Corresponding author.}
\IEEEauthorblockA{\text{$^1$Gaoling School of Artificial Intelligence, Renmin University of China, Beijing, China}  \\
	\text{$^2$Beijing Key Laboratory of Big Data Management and Analysis Methods, Beijing, China} \\
	\text{$^3$Noah's Ark Lab, Huawei, Beijing, China}\\
		wang.zhenlei@foxmail.com, successcx@gmail.com, ruizhou@ruc.edu.cn, daiquanyu@huawei.com, \\dongzhenhua@huawei.com, jrwen@ruc.edu.cn}
}
\maketitle

\begin{abstract}
The key of sequential recommendation lies in the accurate item correlation modeling.
Previous models infer such information based on item co-occurrences, which may fail to capture the real causal relations, and impact the recommendation performance and explainability.
In this paper, we equip sequential recommendation with a novel causal discovery module to capture causalities among user behaviors.
Our general idea is firstly assuming a causal graph underlying item correlations, and then we learn the causal graph jointly with the sequential recommender model by fitting the real user behavior data.
More specifically, in order to satisfy the causality requirement, the causal graph is regularized by a differentiable directed acyclic constraint.
Considering that the number of items in recommender systems can be very large, we represent different items with a unified set of latent clusters, and the causal graph is defined on the cluster level, which enhances the model scalability and robustness.
In addition, we provide theoretical analysis on the identifiability of the learned causal graph.
To the best of our knowledge, this paper makes a first step towards combining sequential recommendation with causal discovery.
For evaluating the recommendation performance, we implement our framework with different neural sequential architectures, and compare them with many state-of-the-art methods based on real-world datasets. Empirical studies manifest that our model can on average improve the performance by about {6.1\% and 11.3\%} on $F_1$ and NDCG, respectively.
To evaluate the model explainability, we build a new dataset with human labeled explanations for both quantitative and qualitative analysis.
\end{abstract}


\begin{figure}[t]
	\centering
	\setlength{\fboxrule}{0.pt}
	\setlength{\fboxsep}{0.pt}
	\fbox{
		\includegraphics[width=.97\linewidth]{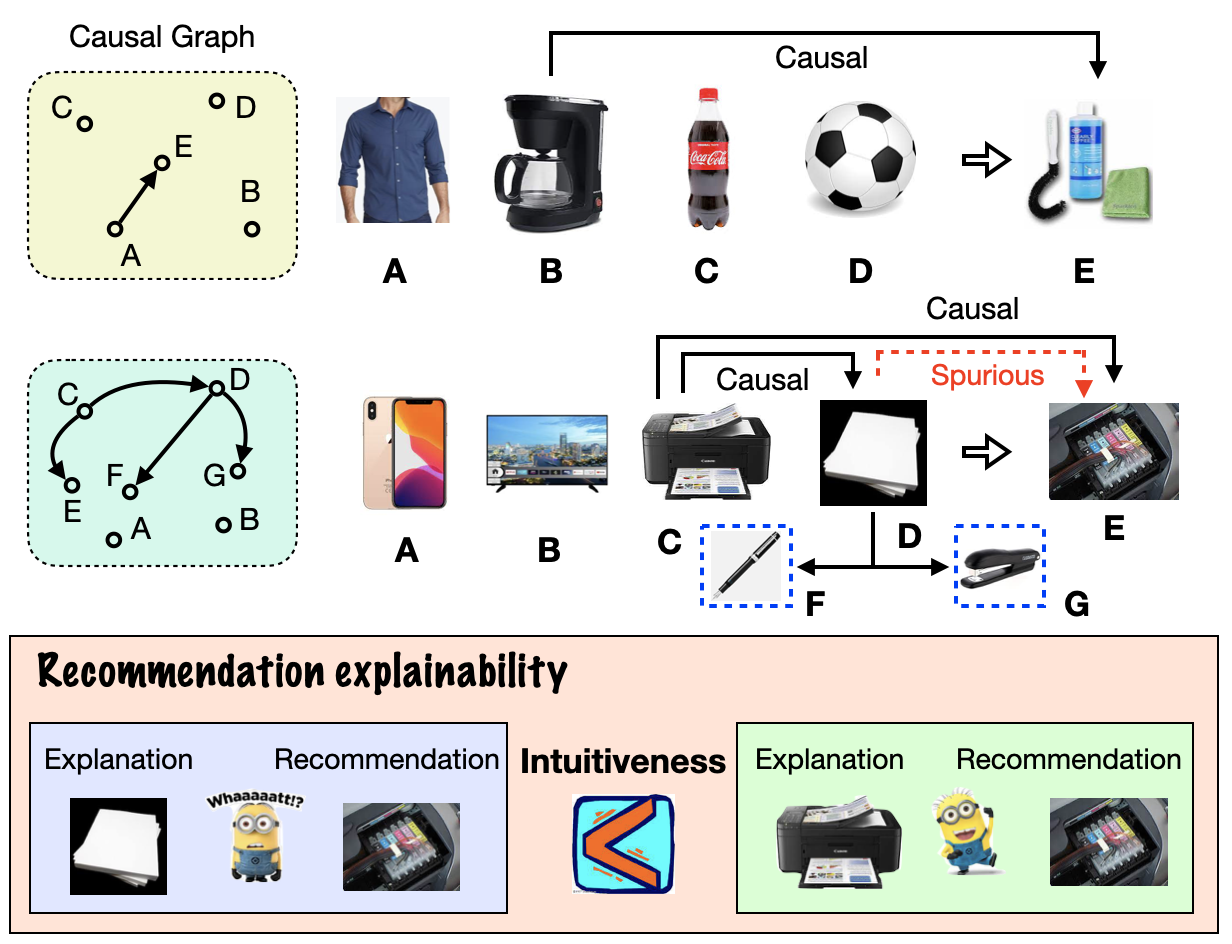}
	}
	\vspace*{-0.1cm}
	\caption{
	{Illustrating our motivation on why modeling causalities among user behaviors is important.
		In the first example, the causally irrelevant items such as the T-shirt and football may confuse the prediction of the pot cleaner, thus removing them may lead to better recommendation performance.
		In the second example, although the paper and ink box are always observed in the same user behavior sequence, they cannot causally influence each other. 
        Basically, their co-occurrence is because they are both causally trigged by the printer.}
	}
	\label{intro}
	\vspace*{-0.2cm}
\end{figure}

\section{Introduction}\label{introduction}
Sequential recommendation is motivated by the fact that user behaviors are always correlated.
For example, in Figure~\ref{intro}, the interaction with a coffee pot may lead to the purchasing of a pot cleaner.
A user buys an ink box because she has previously interacted with a printer.
Such correlations are key to sequential recommendation, and people have introduced different assumptions to model them effectively.
For example, FPMC~\cite{rendle2010factorizing} regards user behaviors as Markov chains, where the current behavior is only influenced by the most recent action.
GRU4Rec~\cite{hidasi2015session} relaxes the Markov assumption, and predicts behaviors by taking more history information into consideration.
STAMP~\cite{liu2018stamp} leverages the attention mechanism to discriminate the importances of different history items.

{Psychological research shows that human behaviors always follow some causal patterns~\cite{rychlak1994logical,rychlak1986logical}.}
For example, a user buys an ink box typically because she has purchased a printer before.
While existing sequential recommender models have achieved many successes, there is no mechanism to capture such essential causalities among user behaviors.
The advantages of causality modeling mainly lie in two folds:
on the one hand, causality is an essential nature on human behavior correlations.
It is stable and intrinsic, which facilitates more informative and robust model learning. 
As exampled in Figure~\ref{intro}, the ink box is caused by the printer, other causally irrelevant items like the phone can be redundant or even confused for predicting the target item.
Notably, there are many attention-based models~\cite{li2017neural,liu2018stamp} for discriminating item correlations.
However, the captured relationships may not well reflect causalities.
In Figure~\ref{intro}, the purchasing of a printer may usually trigger the interactions with the paper and ink box.
Thus, one can always simultaneously observe the paper and ink box in the same user behavior sequences, which makes the attention weight between them very high.
However, these items are causally irrelevant, they are both the results of the printer.
Without the printer, the following item of the paper is more likely to be a pen or stapler, which means the paper is not a good indicator of the ink box.
In this example, the printer is the only direct and reliable cause of the ink box, which is context-invariant.
On the other hand, causality is always aligned with human intuitions, which facilitates more accessible recommendation explanations.
Previous explainable strategies are mostly based on the attention mechanism, which may sometimes lead to unreasonable results.
Also see the example in Figure~\ref{intro}, the attention weights between the paper and ink box can be very high.
However, explaining them with each other is less intuitive, which may confuse the users or even bias their decisions.
Motivated by the above considerations, we would like to ask: ``Can we explicitly capture the causalities among user behaviors when building sequential recommender models?''.

To answer this question, in this paper, we propose a \underline{\textbf{cau}}sality enhanced \underline{\textbf{se}}quential \underline{\textbf{r}}ecommendation framework (called \textbf{Causer} for short).
Our general idea is firstly assuming a causal graph underlying item correlations.
Then, for each item, its causally irrelevant history information is filtered out for deriving more reasonable representations.
Instead of handcrafting the causal graph, we learn it in the optimization process.
The optimal causal graph is expected to well purify the history information, such that user behaviors can be better fitted and explained.
While this idea seems to be feasible, it is non-trivial due to the following difficulties:
(1) to begin with, the number of items can be very large.
Directly building item-level causal graph is intractable and hard to scale.
How to handle the large item space may pose great challenges to realize our idea.
(2) Then there are {few works} on studying user sequential behaviors from the causal discovery perspective.
How to learn the causal graph, and integrate it into sequential recommender models are still open problems.
(3) At last, {even if we can design a causal sequential model}, whether the learned causal graph can be correctly identified is still not clear, which may challenge our idea from the theoretical perspective.

In order to overcome these difficulties, we represent each item by a mixture of the latent clusters based on an encoder-decoder architecture.
Then, we define the causal graph on the cluster level, which greatly reduces the graph size, and thus can be more scalable (corresponding to the above point (1)).
In order to make the causal graph learnable, we apply the idea of NOTEARS~\cite{zheng2018dags} to sequential recommender models, where the causal graph is continuelized and regularized by a directed acyclic constraint (corresponding to the above point (2)).
At last, we theoretically {prove} that the causal graph learned in our framework is identifiable to the true Markov equivalent class under mild assumptions (corresponding to the above point (3)).

In a summary, the main contributions of this paper can be concluded as follows:

$\bullet$ We propose to enhance sequential recommendation by capturing causal relations among user behaviors, which, to the best of our knowledge, is the first time in the recommendation domain.

$\bullet$ To realize the above idea, we design a general framework by infusing a causal discovery module into sequential recommender models.
To make this framework more scalable, we further propose an encoder-decoder architecture to define the causal graph on the cluster-level.

$\bullet$ We provide theoretical analysis on the identifiability of the causal graph learned in our framework.

$\bullet$ We conduct extensive experiments to demonstrate the effectiveness of our framework in promoting the recommendation performance.

$\bullet$ For evaluating the recommendation explainability, we build a dataset by manually labeling the causal relations between different items.

In the following sections, we firstly introduce the preliminaries of this paper in section~\ref{pre}, and then we detail our framework, and present the rationalities of our model designs in section~\ref{cau-model}. In the next, we review the previous work, which are related with our studies in section~\ref{rela}. To demonstrate the effectiveness of our model, we conduct extensive experiments in section~\ref{exper}. The conclusion and outlook come at last in section~\ref{conclu}.

\section{Preliminaries}\label{pre}
\subsection{Sequential Recommendation}
In sequential recommendation, the current user preference is predicted based on the history information.
Suppose we have a user set $\mathcal{U}$ and an item set $\mathcal{V}$.
The interactions\footnote{Here, ``interaction'' is an umbrella term, which can be click, purchase and so on, {and we assume that all the interactions are logged in the dataset.}} between the users and items are chronologically organized into a set $\mathcal{O}=\{(u_k,\vec{\bm{v}}_k^1,\vec{\bm{v}}_k^2,...,\vec{\bm{v}}_k^{l_k})\}_{k=1}^N$, where in each element, $u_k\in \mathcal{U}$ is a user,  $\vec{\bm{v}}_k^{j}\subset \mathcal{V}$ is an item set represented by a $|\mathcal{V}|$-dimensional 0-1 vector.
{$j=1,...,l_k$ corresponds to increasing timestamps.}
$\vec{\bm{v}}_k^1,\vec{\bm{v}}_k^2,...,\vec{\bm{v}}_k^{l_k}$ are sequentially interacted by user $u_k$, and $l_k$ is the sequence length.
This formulation is generally compatible with the ordinary sequential recommendation~\cite{li2017neural} and the next basket recommendation~\cite{rendle2010factorizing}, where for the former case, there is only one ``1'' in $\vec{\bm{v}}_k^{j}$, and in the latter task, $\vec{\bm{v}}_k^{j}$ is a multi-hot vector.
In the optimization process, the log likelihood of observing $(\vec{\bm{v}}_k^1,\vec{\bm{v}}_k^2,...,\vec{\bm{v}}_k^{l_k})$ given $u_k$ is:
	\begin{eqnarray}\label{or-loss}
		\begin{aligned}
			\log{p(\vec{\bm{v}}_k^1,\vec{\bm{v}}_k^2,...,\vec{\bm{v}}_k^{l_k}|u_k)} &= \sum_{j=1}^{l_k} {\log{f(\vec{\bm{v}}_{k}^{j}|\bm{H}_{k}^j, {u_k})}}\\
			&= \sum_{j=1}^{l_k}\sum_{b=1}^{|\mathcal{V}|}{\log{f([\vec{\bm{v}}_{k}^{j}]_b|\bm{H}_{k}^j, {u_k})}},
		\end{aligned}
	\end{eqnarray}
where 
$f$ can be any sequential architecture like LSTM~\cite{li2017neural} or GRU~\cite{hidasi2015session}.
$\bm{H}_{k}^j = \{\vec{\bm{v}}_k^1,\vec{\bm{v}}_k^2,...,\vec{\bm{v}}_k^{j-1}\}$ is the history information.
$[\vec{\bm{v}}_{k}^{j}]_b\in \{0,1\}$ is the $b$th element in $\vec{\bm{v}}_{k}^{j}$, indicating whether item $b$ is interacted at step $j$.
Straightforwardly, one can directly optimize $f(\vec{\bm{v}}_{k}^{j}|\bm{H}_{k}^j, {u_k})$ by deploying a {softmax} output layer on $f$.
However, this method can be less effective due to the large number of items.
{In practice, people usually leverage the sigmoid function~\cite{li2017neural} to predict the interaction probability of each item separately, where one can adopt negative sampling~\cite{li2017neural,liu2018stamp} to speed up the training process.}
At last, the parameters of $f$ are optimized by maximizing the total log likelihood of all the training samples, that is, 
\begin{eqnarray}
\begin{aligned}
\sum_{k=1}^N{\log{p(\vec{\bm{v}}_k^1,\vec{\bm{v}}_k^2,...,\vec{\bm{v}}_k^{l_k}|u_k)}},
\end{aligned}
\end{eqnarray}
where N is the total number of user behavior sequences.

\setlength{\textfloatsep}{0.1cm}
\begin{table}[t]
	\centering
	\caption{{{Notations used in this paper}}}
	\vspace{-0.cm}
	\renewcommand\arraystretch{1.3}
	\scalebox{1.}{
		\begin{threeparttable} 
			\begin{tabular}{p{3.cm}<{\centering}|p{4.8cm}<{\centering}}
				\hline\hline
				Notation  & Description \\ \hline
				$\mathcal{U}$  &The user set in our problem.\\ 
				$|\mathcal{U}|$  &The number of users.\\ 
				$u_k$  &A user in $\mathcal{U}$. \\ 
				$\mathcal{V}$  &The item set in our problem.\\ 
				$|\mathcal{V}|$  &The number of items.\\ 
				$\vec{\bm{v}}_k^{j}$ &The $j$th item set in the $k$th interaction sequence, which is represented by a $|\mathcal{V}|$-dimensional multi-hot vector \\ 
				$[\vec{\bm{v}}_{k}^{j}]_b$  &the $b$th element in $\vec{\bm{v}}_{k}^{j}$ \\ 
				$\bm{H}_{k}^j $  & The history interactions before time step j.\\ 
				$\mathcal{O}$  &The interaction set between the users and items in our problem.\\ 
				$N$  &The number of interaction sequences. \\ 
				$l_k$  &the length of the \emph{k}-th interaction sequence.\\ 
				$\bm{W}\in\mathbb{R}^{|\mathcal{V}|\times |\mathcal{V}|}$& The item-level causal relation matrix.\\ 
	            $\bm{L}_b^j(\bm{W})$&The filtered history information for predicting item $b$. \\ 
	            $\bm{W}_{.b}$&The $b$th column of $\bm{W}$.\\ 
				$K$  & The number of item clusters \\ 
				$\tilde{\bm{v}}\in \mathbb{R}^{d}$  & The raw features of item $v$, for example, the textual description of the item. \\ 
				$\overline{\bm{v}}\in \mathbb{R}^{K}$  & The cluster assignment vector of item $v$.\\ 
				$m_k \in \mathbb{R}^{d_2}$  & The $k$th cluster center vector.\\ 
                $\bm{V}_1, \bm{V}_2, \bm{b}_1, \bm{b}_2$& The parameters in the encoder function.\\ 
				$\bm{W}^c\in \mathbb{R}^{K\times K}$  &The cluster-level causal relation matrix.\\ 
                $\bm{\alpha}$&The free parameters for learning $\overline{\bm{v}}$.\\ 
                $\eta$&The {softmax} temperature parameter.\\ 
                $\bm{V}_3, \bm{V}_4, \bm{b}_3, \bm{b}_4$& The parameters in the decoder function.\\ 
                $\bm{V}\in \mathbb{R}^{d_e\times d_h}$&The weighting matrix for adapting the embedding and hidden state spaces.\\ 
                $\hat{{W}}_{\vec{\bm{v}}_k^{t}b}$ &The total causal effect from the items in $\vec{\bm{v}}_k^{t}$ to item $b$.\\ 
                $\alpha_t$&The attention weight at step $t$.\\ 
                $\bm{A}$&The projection parameter in the attention network.\\ 
                $\beta_1, \beta_2$ &The Lagrange multiplier, and the penalty parameter.\\ 
                $\Theta_{g}$ &The parameters in the sequential model $g$.\\ 
                $\Theta_{e}$ &The set of user/item embedding parameters.\\ 
                $\Theta_{a}$ &The parameters in~(\ref{deepclu}) and~(\ref{dec}).\\ \hline
				\hline
			\end{tabular}
	\end{threeparttable} }
	\label{rec-notation}
	\vspace{-0.cm}
\end{table}

\subsection{Causal Discovery}
Causal relation refers to the concept that one event will result in the occurrence of the other ones.
For example, buying a printer may lead to the interaction with the ink box.
Causal discovery aims to infer causal relations from the observation data.
Formally, suppose we have a set of random variables $\bm{X} = \{X^1,X^2,...,X^m\}$.
Their causal relations are determined by a causal graph $G$, where the adjacency matrix $\bm{W}\in\mathbb{R}^{m\times m}$ is defined in the following way:
if $X^i$ is the cause of $X^j$, then $W_{ij} = 1$, otherwise, $W_{ij} = 0$.
{When $W_{ij} = 1$, it means that the change of $X^i$ will cause the change of $X^j$.}
{In the recommendation domain, if item A can cause item B, then observing A will lead to the observation of item B, not observing A will result in the absent of item B.
In many cases, we may observe interchangeable orders between two items. 
However, we cannot easily say that these items are causally relevant with each other.
Determining the their causal relations may depend on whether the absent of one item can lead to the disappearance of the other one.
}

Given an observed dataset $\mathcal{T}=\{(x_i^1,x_i^2,...,x_i^m)\}_{i=1}^n$, where each sample $(x_i^1,x_i^2,...,x_i^m)$ is an implementation of $\bm{X}$.
Causal discovery aims to learn $W$ based on $\mathcal{T}$.
Unlike ordinary structure learning, causal discovery requires $\bm{W}$ to be directed acyclic, which encodes the fact that the cause and effect are not commutable, e.g., if $X^i$ causes $X^j$, then $X^j$ cannot be the reason of $X^i$.
{
It should be noted that the acyclic requirement is the basic nature of causality.
In the recommendation domain, while one may observe interchangeable orders between two items, it does not mean the causal relation can be cyclic.
They may be both causally triggered by the other factors.
For the classical ``beer and diaper'' example, the two items are frequently observed together and their orders can be interchangeable. However, there is no causal relation between them, and it is unreasonable to explain them by each other.
}

In order to solve the causal discovery problem, people have designed quite a lot of methods~\cite{ng2019graph,zheng2018dags,brouillard2020differentiable}, among which NOTEARS~\cite{zheng2018dags} is a very popular one due to its differentiable nature.
Basically, NOTEARS aims to solve the following optimization problem
	\begin{eqnarray}\label{op}
		\begin{aligned}
			&\min_{\bm{W}}\sum_{i=1}^n\sum_{j=1}^m(x_i^j-\bm{x}_i^T\bm{W}_{\cdot j})^2 + \lambda||\bm{W}||_1,\\
			&s.t.~~~~~~~\text{trace}(e^{\bm{W}\odot \bm{W}}) = m,
		\end{aligned}
	\end{eqnarray}
where $W_{\cdot j}\in\mathbb{R}^{m}$ is the $j$th column of $\bm{W}$, $\bm{x}_i=(x_i^1,x_i^2,...,x_i^m)\in\mathbb{R}^{m}$ is the $i$th sample.
Multiplying $\bm{x}_i$ with $W_{\cdot j}$ actually means regressing $x_i^j$ by all its cause variables, since the irrelevant variables are filtered out by $W_{ij} = 0$.
$\odot$ is element-wise multiplication. The constraint $\text{trace}(e^{\bm{W}\odot \bm{W}}) = m$ aims to ensure that $\bm{W}$ is a directed acyclic graph (DAG).
Let $S=\bm{W}\odot \bm{W}$, then according to taylor expansion, $\text{trace}(e^{S}) = \text{trace}(\bm{I})+\text{trace}(S)+\text{trace}(S^2)+...$,
recall that $[S^k]_{ij}$ is the number of k-step paths from $X^i$ to $X^j$, then $\text{trace}(e^{S})=m \Rightarrow \text{trace}(S)+\text{trace}(S^2)+...=0$,
which means there is no path from $X^i$ to $X^i$ with any-steps.

\noindent
\textbf{Markov equivalent class (MEC).}
Ideally, the learned causal graph should be exactly aligned with the true causal relations among different variables.
However, it is impossible to directly derive the true causal graph without enough assumptions~\cite{brouillard2020differentiable}.
In practice, Markov equivalent class (MEC)~\cite{brouillard2020differentiable} is usually leveraged to verify whether the learned causal graph is satisfied, which is defined as follows:
\begin{definition}
	Two DAGs $G_1$ and $G_2$ are said to be in the same Markov equivalent class if they share the same skeleton and v-structures, that is:
	(1) {for any two nodes $i$ and $j$, if there is an edge between $i$ and $j$ in $G_1$, then $i$ and $j$ are also directly connected in $G_2$.}
	(2) {for any three nodes $i$, $j$ and $k$, if their relations in $G_1$ are $i\rightarrow k\leftarrow j$\footnote{Here, $\rightarrow$ and $\leftarrow$ are direct edges.}, then the same relations are also valid in $G_2$, and vice versa.}
\end{definition}
If the causal graph learned from an algorithm $\mathcal{A}$ falls into the same Markov equivalent class with the true causal graph\footnote{Actually, after obtaining the causal graph in the true MEC, little efforts are needed to determine the true causal graph~\cite{brouillard2020differentiable}.}, then we say the causal graph is \textit{identifiable} by $\mathcal{A}$.
For more technique details about causal discovery, we refer the readers to~\cite{brouillard2020differentiable} for more comprehensive introduction.
In this paper, we combine causal discovery with sequential recommendation to capture causalities among user behaviors, which, to the best of our knowledge, is the first time in the recommendation domain.
For clear presentation, we list the notations leveraged throughout this paper in Table ~\ref{rec-notation}

\section{The Causer Model}\label{cau-model}
In this section, we describe our framework more in detail.
Formally, suppose the underlying causal relation between different items is defined by $\bm{W}\in\mathbb{R}^{|\mathcal{V}|\times |\mathcal{V}|}$, where if item $i$ is the cause of item $j$, then $W_{ij}=1$, otherwise, $W_{ij}=0$.
We denote by $\bm{W}_{.b}$ the $b$th column of $\bm{W}$, indicating all the causes of item $b$.
Given a user behavior sequence $(u,\vec{\bm{v}}^1,\vec{\bm{v}}^2,...,\vec{\bm{v}}^{l})$, we firstly compute its log likelihood based on $\bm{W}$, which improves equation~(\ref{or-loss}) to the following objective: 
	\begin{eqnarray}\label{imp-loss}
		\begin{aligned}
			&\log{p(\vec{\bm{v}}^1,\vec{\bm{v}}^2,...,\vec{\bm{v}}^{l}|u)} =\\ &\sum_{j=1}^{l}\sum_{b=1}^{|\mathcal{V}|}{\log{f([\vec{\bm{v}}^{j}]_b|\bm{L}_b^j(\bm{W}), {u})}},
		\end{aligned}
	\end{eqnarray}
where $\bm{L}_b^j(\bm{W}) = \{\vec{\bm{v}}^1\odot\bm{W}_{.b},\vec{\bm{v}}^2\odot\bm{W}_{.b},...,\vec{\bm{v}}^{j-1}\odot\bm{W}_{.b}\}$.
At each step, the operation $\vec{\bm{v}}^t\odot\bm{W}_{.b}~(t\in[1,j-1])$ aims to remove all the causally irrelevant items, which facilitates more focused history representation.
As mentioned before, causalities can reveal robust and intrinsic item correlations, the purified training instance $\{(\bm{L}^j_b(W), {u}),[\vec{\bm{v}}^{j}]_b\}$ makes it easier to learn the basic user behavior patterns.
Given the training dataset $\mathcal{O}$, the final loss function is:
	\begin{eqnarray}\label{ini-loss}
		\begin{aligned}
			\min &-\sum_{k=1}^N \sum_{j=1}^{l}\sum_{b=1}^{|\mathcal{V}|}{ \log{f([\vec{\bm{v}}^{j}]_b|\bm{L}_b^j(\bm{W}), {u})}}+ \lambda||\bm{W}||_{1}\\
			s.t.&\quad \text{trace}(e^{\bm{W}\odot \bm{W}}) = |\mathcal{V}|
		\end{aligned}
	\end{eqnarray}
where $||\bm{W}||_{1}$ aims to encourage sparse causal relations among items.
$\lambda > 0$ is the regularizer coefficient. 
With a large $\lambda$, the causal graph is regularized to be sparse.
While when $\lambda$ is smaller, the causal graph is allowed to be denser.
The constraint $\text{trace}(e^{\bm{W}\odot \bm{W}}) = |\mathcal{V}|$ ensures that the causal graph is acyclic.

The key to the above idea lies in how to derive $\bm{W}$.
The most straightforward method is manually defining it based on human experiences.
For example, people always buy a pot cleaner after purchasing the coffee pot, and interacting with the ink box is usually because the user has purchased a printer.
While this method is easy to execute, it suffers from several significant weaknesses:
to begin with, handcrafting the causal relations for each item pair is too labor intensive, where one needs to label $\mathcal{O}(|\mathcal{V}|^2)$ item pairs!
In addition, the manually defined $\bm{W}$ is hard to generalize, that is, one cannot infer causal relations among newly appeared items, which makes it difficult to be applied in practice.

In order to overcome these weaknesses, and inspired by the recent advances~\cite{zheng2018dags,brouillard2020differentiable,ng2019graph} in causal discovery, we propose to learn $\bm{W}$ adaptively from the data.
However, directly applying causal discovery algorithms like NOTEARS to sequential recommender models is not easy because:
(1) the number of items (i.e., $|\mathcal{V}|$) can be very large, which makes it hard to efficiently store and optimize $\bm{W} (\in \mathbb{R}^{|\mathcal{V}|\times |\mathcal{V}|})$.
(2) Basically, different item correlations may follow some common high-level patterns. For example, the relations of ``coffee pot $\rightarrow$ pot cleaner'' and ``printer $\rightarrow$ ink box'' are both specifications of the pattern ``office/living items $\rightarrow$ accessories''. 
Directly fitting item-level correlations may result in too sensitive models, which fails to capture the above robust underlying patterns.

\vspace{-0.cm}
\subsection{Cluster-level Causal Graph}
For learning causalities in a more feasible manner, we assume that the item space is structured, which can be expanded by a small amount of latent clusters~\cite{lee2013local,haeffele2014structured}.
Instead of building item-level causal graph, we firstly represent each item as a mixture of the latent clusters, and then define the causal graph on the cluster-level, which may facilitate more scalable and robust optimization.

For efficient training, we cluster the items in a fully differentiable manner~\cite{yang2016joint,yang2017towards}.
More specifically, for each item $v$, we firstly project its raw features\footnote{{The raw features can be any information describing the item profiles, such as the item descriptions and so on.}} $\tilde{\bm{v}}\in \mathbb{R}^{d}$ into an embedding as follows:
	\begin{eqnarray}\label{encoder}
		\begin{aligned}
			{\bm{v}^*} = \bm{V}_2\sigma(\bm{\bm{V}_1\tilde{\bm{v}}+\bm{b}_1})+\bm{b}_2,
		\end{aligned}
	\end{eqnarray}
where
$\bm{V}_1\in \mathbb{R}^{d_1\times d}$, $\bm{b}_1\in \mathbb{R}^{d_1}$, $\bm{V}_2\in \mathbb{R}^{d_2\times d_1}$, $\bm{b}_2\in \mathbb{R}^{d_2}$ are weighting parameters. $\sigma(x) = \frac{1}{1+e^{-x}}$ is the sigmoid function.
Suppose we have K clusters and the center of the $k$th cluster is $\bm{m}_k\in \mathbb{R}^{d_2}$.
Let $\overline{\bm{v}}\in \mathbb{R}^{K}$ be a cluster assignment vector, with each element $\overline{{v}}_k$ indicating the probability that item $v$ is assigned to cluster $k$. 
Then, we represent each item with a mixture of the clusters by minimizing the following loss:
	\begin{eqnarray}\label{deepclu}
		\begin{aligned}
			&\min \!\sum_{v\in \mathcal{V}}\!||\bm{v}^*\!-\!\sum_{k=1}^K \overline{{v}}_k\bm{m}_k||^2_2, \\
			&~s.t. \sum_{k=1}^K \overline{{v}}_k\!=\!1,~\overline{{v}}_k\!>\!0~\forall k\!\in\![1,K]
		\end{aligned}
	\end{eqnarray}
where the assignment vector $\overline{\bm{v}}$ and cluster center $\bm{m}_k$ are jointly learned in the optimization process.
By this objective, the item embeddings are constrained to be different convex combinations on a unified set of cluster centers.

It should be noted that both $\bm{v}^*$ and $\overline{\bm{v}}$ are representations of item $v$, but the former aims to reveal semantics, while the latter indicates cluster assignments.
For efficiently optimizing this objective, we introduce a free parameter $\bm{a}\in \mathbb{R}^{K}$ to relax the constraints.
In specific, we let $\overline{{v}}_k = \frac{\exp{({a}_k/\eta}) }{\sum_{k=1}^{K} \exp{(a_k/\eta) }}$,
which makes the constraints $\sum_{k=1}^K \overline{{v}}_k\!=\!1,~\overline{{v}}_k\!>\!0~\forall k\!\in\![1,K]$ always hold for any $\bm{a}$.
The temperature $\eta$ is leveraged to tune the hardness of the assignment vector.
If $\eta\rightarrow 0$, the assignment vector is one-hot. As $\eta$ becomes larger, the cluster distribution becomes disperser. 

At last, we feed the item embedding $\bm{v}^*$ into a decoder, and expect to reconstruct the item raw features $\tilde{\bm{v}}$ from its output, that is:
	\begin{eqnarray}\label{dec}
		\begin{aligned}
		\min \sum_{v\in \mathcal{V}}||\hat{\bm{v}}-\tilde{\bm{v}}||^2_2,
		\end{aligned}
	\end{eqnarray}
where the decoder for reconstruction is specified as: $\hat{\bm{v}} = \bm{V}_4\sigma(\bm{\bm{V}_3\bm{v}^*+\bm{b}_3})+\bm{b}_4$,
We denote by $\Theta_{a}$ all the parameters in~(\ref{deepclu}) and~(\ref{dec}), that is, $\Theta_{a} = \{\bm{V}_1,\bm{V}_2,\bm{V}_3,\bm{V}_4,\{\bm{m}_k\}_{k=1}^K,\overline{\bm{v}},\bm{b}_1,\bm{b}_2,\bm{b}_3,\bm{b}_4\}$.

Based on the above learning objectives, we can obtain a cluster assignment vector for each item.
We define by $\bm{W}^c=[\bm{W}^c_{ij}]\in \mathbb{R}^{K\times K}$ the causal relations among different clusters, where ${W}^c_{ij}=1$ means cluster $i$ is the cause of cluster $j$.
For example, office/living items and accessories are two clusters, and people usually purchase accessories after buying office/living items.
Thus $W^c_{ij}=1$ for the relation from office/living items to accessories.
If the clusters $i$ and $j$ are causally irrelevant, then $W^c_{ij}=W^c_{ji}=0$.
Based on $\bm{W}^c$, the causal relation between two items $a$ and $b$ can be computed as:
{\setlength\abovedisplayskip{2pt}
	\setlength\belowdisplayskip{2pt}
	\begin{eqnarray}\label{st}
		\begin{aligned}
			W_{ab} = \overline{\bm{a}}^T\bm{W}^c\overline{\bm{b}} = \sum_{i=1}^K\sum_{j=1}^K \overline{{a}}_i W^c_{ij} \overline{{b}}_j,
		\end{aligned}
	\end{eqnarray}
}
where $\overline{\bm{a}}$ and $\overline{\bm{b}}$ are cluster assignment vectors of item $a$ and $b$ in equation~(\ref{deepclu}).
In order to compute the item-level causal relation, this equation iterates all the cluster pairs, and for each pair, the cluster-level causal relation is multiplied by the probabilities that the items are assigned to the clusters.
Extremely, if the items are clustered in a hard manner (\emph{i.e.}, $\eta\rightarrow 0$), then the item-level causal relation is exactly the corresponding cluster-level causal relation\footnote{In practice, one can disambiguate item causal relations via controlling $\eta$.}. 
An intuitive example of computing $W_{ab}$ can be seen in Figure~\ref{model}(a)

\begin{figure}[t]
	\centering
	\setlength{\fboxrule}{0.pt}
	\setlength{\fboxsep}{0.pt}
	\fbox{
		\includegraphics[width=.9\linewidth]{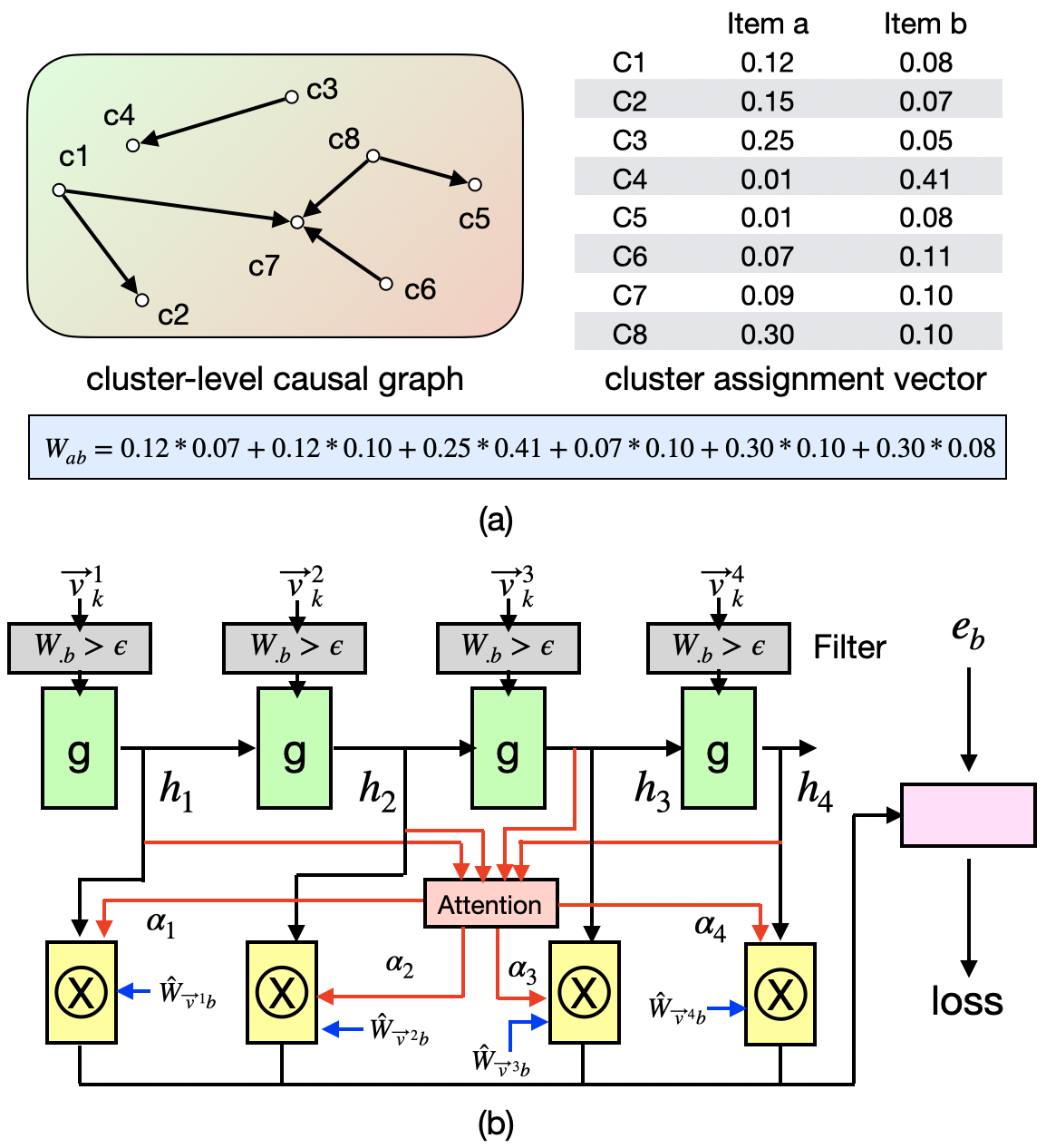}
	}
	\vspace*{-0.2cm}
	\caption{
		(a) The left is an exampled cluster-level causal graph, and the right are two cluster assignment vectors.
		The item-level causal relation is computed at the bottom. 
		(b) Our model architecture, where the red and blue lines represent the local and global item relations, respectively.
	}
	\label{model}
	\vspace*{-0.cm}
\end{figure}

\subsection{Model Implementation}
Given the history information $(u_k,\vec{\bm{v}}_k^1,\vec{\bm{v}}_k^2,...,\vec{\bm{v}}_k^{j-1})$, the probability of interacting with item b at step $j$ is:
	\begin{eqnarray}\label{cx}
		\begin{aligned}
			&\bm{h}_{t+1} = g (\bm{h}_{t}, \vec{\bm{v}}_k^{t}\odot\bm{1}(\bm{W}_{.b}>\epsilon),u_k),~~t\in[1,j-1]\\
			&f([\vec{\bm{v}}_k^{j}]_b=1|\bm{H}_{k}^j, {u}_k)  = \sigma{(\bm{e}_b^T (\bm{V}\sum_{t=1}^{j-1} \hat{{W}}_{\vec{\bm{v}}_k^{t}b}\alpha_t \bm{h}_t))},
		\end{aligned}
	\end{eqnarray}
where $\bm{H}_{k}^j = \{\vec{\bm{v}}_k^1,\vec{\bm{v}}_k^2,...,\vec{\bm{v}}_k^{j-1}\}$. 
$\bm{W}$ is derived based on equation~(\ref{st}).
$\bm{1}(\bm{W}_{.b}>\epsilon)$ is a vector-level indication function, which basically binarizes $\bm{W}_{.b}$ by a threshold parameter $\epsilon$.
$\vec{\bm{v}}_k^{t}\odot\bm{1}(\bm{W}_{.b}>\epsilon)$ removes the items which are causally less relevant.
If there is no cause item at some step, that is, $\vec{\bm{v}}_k^{t}\odot\bm{1}(\bm{W}_{.b}>\epsilon)$ is an all-zero vector, then we directly skip this step.
Our framework can be implemented with either LSTM~\cite{hochreiter1997long} or GRU~\cite{chung2014empirical}, where input item embeddings are derived based on equation~(\ref{encoder}), and we use $g$ to summarize the computational rules at each step.
$\bm{V}\in \mathbb{R}^{d_e\times d_h}$ is a weighting matrix adapting the embedding and hidden state spaces.
$\bm{e}_b\in \mathbb{R}^{d_e}$ is an independent embedding of item $b$ for computing the similarity between the candidate item and the history information.
$\hat{{W}}_{\vec{\bm{v}}_k^{t}b}$ is the total causal effect from the items in $\vec{\bm{v}}_k^{t}$ to item $b$, and we compute it as 
$(\vec{\bm{v}}_k^{t})^T(\bm{W}_{.b}\odot\bm{1}(\bm{W}_{.b}>\epsilon))$, where $\bm{W}_{.b}\odot\bm{1}(\bm{W}_{.b}>\epsilon)$ basically zeroizes the elements in $\bm{W}_{.b}$ smaller than $\epsilon$ (\emph{i.e.}, removing causally irrelevant information).
$\alpha_t$ is the attention weight computed by $\frac{e^{\text{sim}(\bm{h}_{t},\bm{h}_{j-1})}}{\sum_{k=1}^{j-1} e^{\text{sim}(\bm{h}_{k},\bm{h}_{j-1})}}$,
where $\text{sim}(\bm{h}_{t},\bm{h}_{j-1}) = \bm{h}_{t}^TA\bm{h}_{j-1}$ aims to compute the importance of $\vec{\bm{v}}_k^{t}$ in the history information summarized by $\bm{h}_{j-1}$, $\bm{A}\in \mathbb{R}^{d_h\times d_h}$ is a projection parameter.
Note that the attention is applied to the filtered history information, it aims to discriminate the importances of the items which are already the cause of the target item.

\begin{remark}
(1) In the above model, $\hat{{W}}_{\vec{\bm{v}}_k^{t}b}$ defines the item causal relation in a global manner, where the local context is not considered.
$\alpha_t$ compensates such information by normalizing the importances of the causally relevant items in the same sequence.
By multiplying $\hat{{W}}_{\vec{\bm{v}}_k^{t}b}$ with $\alpha_t$, we aim to capture more comprehensive relations between the history information and target item.
(2) {As an extreme case, where we do not have any training data, our framework recommends items following a nearly uniform distribution.
Empirically, we found that suppose there are N items in the system, then each item is recommended with the probability of $\frac{1}{N}$.}
(3) We illustrate the complete architecture of our model in Figure~\ref{model}(b).
\end{remark}

\subsection{Model Optimization}
For the training set $\mathcal{O}=\{(u_k,v_k^1,v_k^2,...,v_k^{l_k})\}_{k=1}^N$, the parameters of our model are learned by the following optimization problem:
	\begin{eqnarray}\label{ourst}
		\begin{aligned}
			\min &-\sum_{k=1}^N\sum_{j=1}^{l_k} \sum_{b=1}^{|\mathcal{V}|}\{ [\vec{\bm{v}}_k^{j}]_b\log f([\vec{\bm{v}}_k^{j}]_b=1|\bm{H}_{k}^j, {u}_k) \\
			&+(1-[\vec{\bm{v}}_k^{j}]_b)\log(1-f([\vec{\bm{v}}_k^{j}]_b=1)|\bm{H}_{k}^j, {u}_k) \}\\
			&+ \lambda||\bm{W}^c||_{1}+\sum_{v\in \mathcal{V}}||\hat{\bm{v}}-\tilde{\bm{v}}||^2_2\\
			&+ \sum_{v\in \mathcal{V}}||\bm{v}^*-\sum_{k=1}^K \overline{{v}}_k\bm{m}_k||^2_2\\
			s.t.&\quad \text{trace}(e^{\bm{W}^c\odot \bm{W}^c}) = K
		\end{aligned}
	\end{eqnarray}
where we rewrite objective~(\ref{or-loss}) based on equation~(\ref{cx}) by a binary cross-entropy loss.
This optimization problem advances previous sequential recommender models by causally purifying the history information and incorporating causalities when merging the hidden states. 
These abilities are empowered by introducing the cluster-level causal graph $\bm{W}^c$.
Similar to the previous work~\cite{zheng2018dags}, we use the augmented Lagrangian method to solve problem~(\ref{ourst}), which leads to the following loss:
	\begin{eqnarray}
		\begin{aligned}
			&L_1(\Theta_{g}, \Theta_{e}, \Theta_{a}, \bm{V}, \bm{W}^c, \beta_1, \beta_2) = \\
			-&\sum_{k=1}^N\sum_{j=1}^{l_k} \sum_{b=1}^{|\mathcal{V}|}\{ [\vec{\bm{v}}_k^{j}]_b\log f([\vec{\bm{v}}_k^{j}]_b=1|\bm{H}_{k}^j, {u}_k)+ \\
			&(1-[\vec{\bm{v}}_k^{j}]_b)\log(1-f([\vec{\bm{v}}_k^{j}]_b=1)|\bm{H}_{k}^j, {u}_k) \}+\lambda||\bm{W}^c||_{1}\\
			&+\sum_{v\in \mathcal{V}}||\hat{\bm{v}}-\tilde{\bm{v}}||^2_2+\sum_{v\in \mathcal{V}}||\bm{v}^*-\sum_{k=1}^K \overline{{v}}_k\bm{m}_k||^2_2+\beta_1 b(\bm{W}^c) \\
			&+ \frac{\beta_2}{2}|b(\bm{W}^c)|^2\\ \notag
		\end{aligned}
	\end{eqnarray}
where $b(\bm{W}^c) = \text{trace}(e^{\bm{W}^c\odot \bm{W}^c}) - K$.
$\Theta_{g}$ collects all the parameters in the sequential model $g$.
$\Theta_{e}$ is the set of user/item embedding parameters.
$\Theta_{a}$ is the parameter set in~(\ref{deepclu}) and~(\ref{dec}).
$\beta_1$ is the Lagrange multiplier, and $\beta_2>0$ is a penalty parameter.

We summarize the complete optimization process of our model in Algorithm~\ref{random-alg}.
In each training epoch, the item causal relations $\bm{W}$ are firstly derived based on $\bm{W}^c$ and equation~(\ref{st}) (\textit{line 7}).
Then $\bm{W}$ is leveraged to filter the history information (\textit{line 8}). 
At last, the model parameters are alternatively updated until convergence (\textit{line 11-15}).
In the testing phase, the history information is firstly filtered by $\bm{W}$, and then fed into $f$ to predict the target item.

\begin{algorithm}[t] 
	\caption{Learning Algorithm of Causer} 
	\label{random-alg} 
	Initialize the parameters $\Theta_{g}, \Theta_{e}, \Theta_{a}, \bm{V}, \bm{W}^c, \beta_1, \beta_2$.\\
	Indicate the epoch number $N_e$.\\
	Indicate the iteration number $N_i$.\\
	Indicate the threshold parameter $\epsilon$.\\
	Indicate hyper-parameters $\kappa_1>1$ and $\kappa_2<1$.\\
	\For{i in [0, $N_e$]}{
		Compute $\bm{W}$ based on $\bm{W}^c$ and equation~(\ref{st}).\\
		Leverage $\bm{W}$ and $\epsilon$ to filter the history information.\\
		\For{i in [0, $N_i$]}{
			$\bm{W}^{c-} \leftarrow \bm{W}^c$\\
			\For{$\Omega$ in $\{\Theta_{g}, \Theta_{e}, \Theta_{a}, \bm{V}, \bm{W}^c\}$}{
				$\Omega \leftarrow \Omega - \gamma \frac{\partial L_1(\Theta_{g}, \Theta_{e}, \Theta_{a}, \Theta_{c}, \bm{V}, \bm{W}^c, \beta_1, \beta_2)}{\partial \Omega}$\\
			}
			$\beta_1 \leftarrow \beta_1 + \beta_2 b(\bm{W}^c)$\\
			$\beta_2 \leftarrow \kappa_1 \beta_2 \quad if~~|b(\bm{W}^c)|\geq \kappa_2|b(\bm{W}^{c-})|$\\
		}
	}
\end{algorithm}

{
\textbf{Model Efficiency}.
Comparing with the previous sequential recommender models, the most significant part of our framework is introducing the causal matrix $\bm{W}$.
While it can help to learn item causal relations, the model efficiency may also be sacrificed.
In practice, a potential solution for improving the efficiency can be lowering the updating frequency of the parameters irrelevant with the recommender model, that is, we can separate the parameters in line 11 of Algorithm~\ref{random-alg}, and update the $\Theta_{a}$ and $\bm{W}^c$ with a slower pace to reduce the additional cost.
Empirically, we found that if $\Theta_{a}$ and $\bm{W}^c$ are updated every ten epochs, then the training efficiency can be improved by about 22\%. 
In addition, if we can access the prior knowledge on $\bm{{W}}$, then we may pre-train $\bm{{W}}$, and fix it in equation~(\ref{cx}) to improve the training efficiency.
It should be noted that, the inference efficiency of our framework may not be impacted too much, since all the parameters are fixed in this process.
For example, suppose the inference time of SASRec is T, then our framework cost about 1.16T in the testing phase. 
}

\subsection{Identifiability Analysis}
In this section, we analyze the identifiability of the causal graph learned by our framework.
We have the following theory:
\begin{mytheory}
	Let $G^*$ and ${G}$ be the ground truth item-level causal graph and the one learned based on objective~(\ref{ini-loss}).
	Suppose: 
	(\romannumeral1) {the hypotheses class $\bm{F}$ of our model $f$ is large enough to recover the ground truth causal graph.}
	(\romannumeral2) For item sets $A$, $B$ and $C$, if $A$ and $B$ cannot be d-separated by $C$ in the causal graph, then for any $s<t$, $p([\vec{\bm{v}}^t]_A,[\vec{\bm{v}}^t]_B|[\vec{\bm{v}}^{s:1}]_C)\neq p([\vec{\bm{v}}^t]_A|[\vec{\bm{v}}^{s:1}]_C)p([\vec{\bm{v}}^t]_B|[\vec{\bm{v}}^{s:1}]_C)$,
	where $[\cdot]_{O} = \{[\cdot]_i|i\in O\},~~\forall O\in \{A,B,C\}$, and $[\vec{\bm{v}}^{t-1:1}]_K = \{[\vec{\bm{v}}^{t-1}]_K,...,[\vec{\bm{v}}^1]_K\}$.
	(\romannumeral3) {$f([\vec{\bm{v}}_{k}^{j}]_b|\bm{L}_b^j(\bm{W}), {u_k})$ is strictly positive, and the corresponding entropy is finite. }
	Then for small enough $\lambda$, ${G}$ is Markov equivalent to $G^*$.
\end{mytheory}

\begin{proof}
	For easy derivation, we write the negative expectation of loss~(\ref{ini-loss}) as: 
		\begin{eqnarray}\label{noteee}
			\begin{aligned}
				S(G)\!=\!\sup_{\Theta} \sum_{t=1}^{T}\text{E}_{\bm{v}^t \sim p_t(\bm{v}^t)} [\log f_G(\vec{\bm{v}}^t;\Theta)] - \lambda||\bm{W}_G||_{1}\nonumber,
				\end{aligned}
			\end{eqnarray}
	where $$f_G(\vec{\bm{v}}^t;\Theta) = f(\vec{\bm{v}}^t|\bm{L}_b^j(\bm{W}_G);\Theta) = \prod_{b=1}^{|\mathcal{V}|}f([\vec{\bm{v}}^t]_b|\bm{L}_b^j(\bm{W}_G);\Theta),$$
	$\bm{W}_G\in \{0,1\}^{|\mathcal{V}|\times |\mathcal{V}|}$ is the adjacency matrix for a given causal graph $G$. 
	$p_t$ is the data generation probability induced from the ideal causal graph $G^*$ at step $t$. $\Theta$ is the set of model parameters.
	If we can {prove} that the optimal causal graph $G^*$ can lead to larger $S(G^*)$ than the other ones which are not Markov equivalent to $G^*$, then by {maximizing $S(G)$}, we can obtain satisfied causal graphs which are Markov equivalent to $G^*$.
	To begin with, we have:
	{\setlength\abovedisplayskip{3pt}
		\setlength\belowdisplayskip{3pt}
		\begin{eqnarray}\label{ourf}
			\begin{aligned}
				&S(G^*) - S(G)\\
				&=\sup_{\Theta} \sum_{t=1}^T \text{E}_{\vec{\bm{v}}^t \sim p_t}[\log f_{G^*}(\vec{\bm{v}}^t;\Theta)]\\
				&-\! \sup_{\Theta} \!\sum_{t=1}^T \!\text{E}_{\vec{\bm{v}}^t \sim p_t}[\log f_G(\vec{\bm{v}}^t;\Theta)]
			    -\lambda||\bm{W}_{G^*}||_{1}+ \lambda||\bm{W}_G||_{1}\\
				&=-\inf_{\Theta} -\sum_{t=1}^T \text{E}_{\vec{\bm{v}}^t\sim p_t} [\log f_{G^*}(\vec{\bm{v}}^t;\Theta)]- \sum_{t=1}^T\text{E}_{\vec{\bm{v}}^t} [\log p_t(\vec{\bm{v}}^t)]\\
				&+(\inf_{\Theta}-\sum_{t=1}^T \text{E}_{\vec{\bm{v}}^t\sim p_t} [\log f_G(\vec{\bm{v}}^t;\Theta)]+\sum_{t=1}^T \text{E}_{\vec{\bm{v}}^t\sim p_t} [\log p_t(\vec{\bm{v}}^t)])\\
				&- \lambda||\bm{W}_{G^*}||_{1} + \lambda||\bm{W}_G||_{1}\\
				&=\inf_{\Theta} \sum_{t=1}^T KL(p_t(\vec{\bm{v}}^t)||f_G(\vec{\bm{v}}^t;\Theta)) \\
				&-\inf_{\Theta} \sum_{t=1}^T KL(p_t(\vec{\bm{v}}^t)||f_{G^*}(\vec{\bm{v}}^t;\Theta))+ \lambda||\bm{W}_G||_{1}- \lambda||\bm{W}_{G^*}||_{1}\\
				&=\inf_{\Theta} \!\sum_{t=1}^T KL(p_t(\vec{\bm{v}}^t)||f_G(\vec{\bm{v}}^t;\Theta)) \!+\! \lambda(||\bm{W}_G||_{1}\!- \!||\bm{W}_{G^*}||_{1})\\\nonumber
			\end{aligned}
		\end{eqnarray}
	}
	where the last equation holds because of the first assumption, that is, $F$ can perfectly fit the data generation mechanism.
	
	We define $M(G)$ as the set of distributions which are dynamically coherent with graph $G$, that is, for a distribution $f\in M(G)$, if two disjoint node set A and B are d-separated by C in G, then $f([\vec{\bm{v}}^t]_A|[\vec{\bm{v}}^{s:1}]_C)f([\vec{\bm{v}}^t]_B|[\vec{\bm{v}}^{s:1}]_C) = f([\vec{\bm{v}}^t]_A, [\vec{\bm{v}}^t]_B|[\vec{\bm{v}}^{s:1}]_C)$ for any $s<t$. Based on $M(G)$, the following lemma is immediate.
	\begin{lemma}\label{lm}
	  If $p\not \in M(G)$, then $\inf_{f\in M(G)}KL(p||f)>0$.
	\end{lemma}
	According to the definition of $f_{G}$, we have $f_{G}\in M(G)$. Then we consider four cases to demonstrate $S(G^*) - S(G) > 0$, which is similar to~\cite{brouillard2020differentiable}.

	$\bullet$ There is an edge ``$i\rightarrow j$'' in $G^*$, but they are not connected in $G$.
	In this case, $i$ and $j$ are d-separated by $\mathcal{V}\backslash L_{ij}$ in G, where $L_{ij}$ is the set not descendant nodes of $i$ and $j$.
	Because $i$ and $j$ are not d-separated by $\mathcal{V}\backslash L_{ij}$ in $G^*$, they are not independent conditioned on the history $\mathcal{V}\backslash L_{ij}$ based on $p_t$, thus $p_t\not\in M(G)$.
	According to lemma~\ref{lm}, $\inf KL(p_t||f_{G})>0$, thus $\sum_{t=1}^T KL(p_t||f_{G})>0$.
	If $||\bm{W}_G||_{1}>||\bm{W}_{G^*}||_{1}$, then $S(G^*) - S(G) > 0$.
	If $||\bm{W}_G||_{1}<||\bm{W}_{G^*}||_{1}$, then $S(G^*) - S(G) > 0$, when: 
	\begin{eqnarray}\label{note}
		\begin{aligned}
				\lambda< \frac{\inf \sum_{t=1}^T KL(p_t(\vec{\bm{v}}^t)||f_{G}(\vec{\bm{v}}^t;\Theta))}{||\bm{W}_{G^*}||_{1}- ||\bm{W}_{G}||_{1}}.
		\end{aligned}
	\end{eqnarray}
	
    $\bullet$  There is an edge ``$i\rightarrow j$'' in $G$, but they are not connected in $G^*$.
	In this case, if there is an edge in $G^*$ but not in $G$, then we return to (\romannumeral1).
	Thus, we have $||\bm{W}_G||_{1}>||\bm{W}_{G^*}||_{1}$, which lead to $S(G^*) - S(G) > 0$.
	
	$\bullet$  $G$ and $G^*$ share the same skeleton, and there is a v-structure ``$i\rightarrow l \leftarrow j$'' in $G^*$ but not in $G$.
	To begin with, $l\not\in L_{ij}$ in $G$, otherwise there is cycle in the graph.
	As a result, $i$ and $j$ are not independent given $\mathcal{V}\backslash L_{ij}$, since $G$ and $G^*$ share the same skeleton.
	However, $i$ and $j$ are d-separated by $\mathcal{V}\backslash L_{ij}$ in G, thus $p_t\not\in M(G)$, and $S(G^*) - S(G) > 0$.
	
	$\bullet$  $G$ and $G^*$ share the same skeleton, and there is a v-structure ``$i\rightarrow l \leftarrow j$'' in $G$ but not in $G^*$.
	In this case, there must be a path $(i,l,j)$ in $G^*$, which is not d-separated by $\mathcal{V}\backslash L_{ij}$ in $G^*$, since $l\not\in \mathcal{V}\backslash L_{ij}$.
	However, $i$ and $j$ are d-separated by $\mathcal{V}\backslash L_{ij}$ in G, thus $p_t\not\in M(G)$, and $S(G^*) - S(G) > 0$.
   
\end{proof} 

This theory tells us that if we learn our model based on objective~(\ref{ini-loss}), then the obtained causal graph is Markov equivalent to the true causal graph.
It provides theoretical guarantees for our framework, which makes an initial step towards studying the causal identifiability problem in the context of sequential recommendation.

\section{Related Work}\label{rela}
This work stands on the intersection between sequential recommendation and causal inference.
In this section, we briefly introduce the recent advances in these fields and analyze the relations between our framework and these studies.

\textbf{Relation with sequential recommendation.}
Sequential recommendation has recently attracted great interests from both research and industry communities.
Early models like FPMC~\cite{rendle2010factorizing} assumes that user behaviors are only determined by the most recent actions.
Obviously, such Markov assumptions are limited, which can not explore the influence from the behaviors happened longer before. 
With the ever prospering of deep learning techniques, such problem has been alleviated by neural sequential recommender models.
For example, GRU4Rec~\cite{hidasi2015session} leverages recurrent neural network to summarize all the user history behaviors, and HRNN~\cite{quadrana2017personalizing} further extends it with personalized considerations.
When there are multiple history items, an important problem is how to discriminate their importances for the target item.
The above Markov-based models actually have tackled this problem by a heuristic rule, that is, all the history items are not important except the most recent one.
However, such rule may not always hold in practice.
To determine history item importances in a {softer} manner, people have designed quite a lot of attention-based models.
For example, NARM~\cite{li2017neural} combines attention mechanism with gated recurrent unit.
STAMP~\cite{liu2018stamp} uses attention mechanism to separate user long- and short-term engagements.
Bert4Rec~\cite{sun2019bert4rec} leverages self-attention mechanism to capture item long-term dependencies.
Our work continues this research line, focusing on better capturing item correlations.
However, we aim to discover causal relations, which is a significant extension.

\textbf{Relation with causal inference.}
Causal inference stems from applied statistics, and is increasingly leveraged to empower machine learning models.
In general, there are two major tasks in causal inference, that is, causal estimation and causal discovery.
The first problem aims to predict the values of some variables given the causal graph, while the other one aims to learn the causal graph based on a set of observational data.
These two problems are mutually inverse, and our work is more related with causal discovery.
In the past decades, there are mainly two types of causal discovery methods.
The first one is constrain-based, where the causal graph is determined based on conditional independence tests~\cite{spirtes2000causation,meek2013causal,zhang2008completeness}. The second one is score-based, where each causal graph is assigned with a score, and the final result is determined among the graphs with higher scores and satisfying the directed acyclic requirement~\cite{chickering2002optimal,chickering1997efficient,heckerman1995learning,bouckaert1993probabilistic}.
Our work is based on the second method, and we apply it to sequential recommendation, which makes an early attempt on studying causal discovery in user behavior modeling.
Besides, we provide theoretical analysis on the identifiability problem in the context of sequential recommendation.

\section{Experiments}\label{exper}

\subsection{Experiment Setup}

\textbf{Datasets.}
Our experiments are conducted based on the following real-world datasets:
\textbf{Epinions}\footnote{https://cseweb.ucsd.edu/~jmcauley/datasets.html} is a dataset collected from Epinions.com, which includes user ratings and reviews on the products from different categories.
{\textbf{Foursquare}\footnote{https://www.kaggle.com/datasets/chetanism/foursquare-nyc-and-tokyo-checkin-dataset} is a location based recommendation dataset containing user check-ins of the restaurants in Tokyo for about 10 month.}
\textbf{Amazon-Baby}, \textbf{Amazon-Patio} and \textbf{Amazon-Video} are e-commerce datasets\footnote{http://jmcauley.ucsd.edu/data/amazon/}, which are collected from Amazon.com. In these datasets, we have user purchasing records in different product categoties spanning from May 1996 to July 2014. 

For the first four datasets, the item raw features are obtained based on the item descriptions, where each word is represented by an embedding based on GloVe\footnote{https://nlp.stanford.edu/projects/glove/}, and then all the word embeddings are averaged to derive the final item raw features.
For the dataset of Foursquare, the raw features are derived based on the GPS coordinates of the check-in place.
{The basic statistics of our datasets are summarized in Table~\ref{rec-dataset}, and we plot the distribution of the sequence length for each dataset in Fig.~\ref{fig_data}.
We can see all the datasets are extremely sparse which demonstrates the high challenge of our task.
Besides, the number of interactions varies a lot across different datasets, which can help to verify the generality of our model for different data characters.
}

\begin{figure}[t]
	\centering
	\setlength{\fboxrule}{0.pt}
	\setlength{\fboxsep}{0.pt}
	\fbox{
		\includegraphics[width=1.\linewidth]{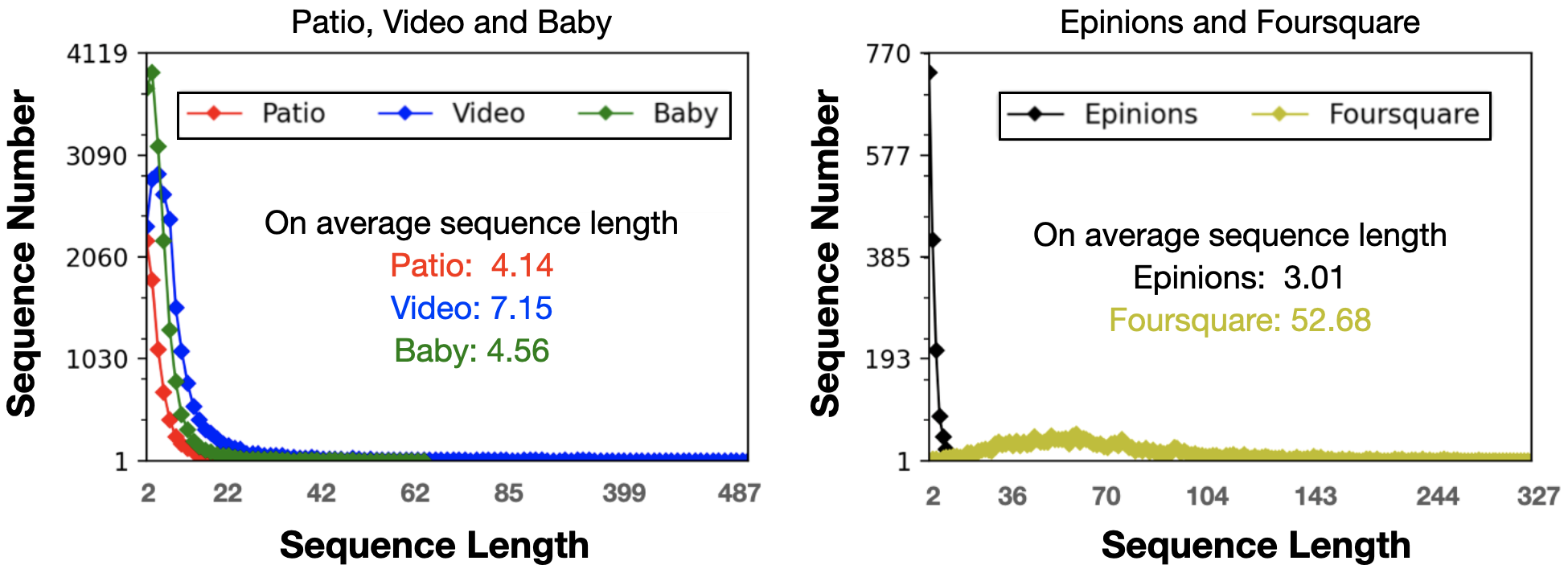}
	}
	\vspace{-0.3cm}
	\caption{{{Distributions of the sequence length for each dataset.
	For saving the space, we merge the statistics of Patio, Baby and Video in the left figure, while the other datasets are summarized in the right one.}}}
	\label{fig_data}
	\vspace{-0.2cm}
\end{figure}

\begin{table}[t]
	\centering
	\caption{{Statistics of the datasets, where the prefix ``Amazon-'' is omitted, and ``SeqLen'' is the average sequence length per user.}}
	\vspace{-0.cm}
	\renewcommand\arraystretch{1.2}
	\scalebox{1.}{
		\begin{threeparttable} 
			\begin{tabular}{p{1.1cm}<{\centering}|p{.8cm}<{\centering}|p{.8cm}<{\centering}|p{1.5cm}<{\centering}|p{1.1cm}<{\centering}|p{.9cm}<{\centering}}
				\hline\hline
				Dataset  &\# User   &\# Item  &\# Interaction&\# SeqLen &Sparsity  \\ \hline
				{Epinions}&1,530&683&4,600&3.01 &99.56\%\\
				{Foursquare}&2,292&5,494&120736&52.68&99.04\%\\
				{Patio}&7,153&2,952&29,625&4.14&99.86\%\\
				{Baby}&16,898&6,178&77,046&4.56&99.93\%\\
				{Video}&19,939&9,275&142,658&7.15&99.92\%\\\hline\hline
			\end{tabular}
	\end{threeparttable} }
	\label{rec-dataset}
	\vspace{-0.cm}
\end{table}

\noindent
\textbf{Baselines.}
We compare our model with the following representative baselines:
\textbf{BPR}~\cite{BPR} is a {well-known} recommender model for capturing user implicit feedback, where the prediction model is specified as matrix factorization.
\textbf{NCF}~\cite{ncf} is a neural recommender model, where the backbone is a combination between two types of generalized matrix factorization models.
\textbf{GRU4Rec}~\cite{hidasi2015session} is a sequential recommender model based on the gated recurrent unit, where each history item is regarded as the input of each step.
\textbf{NARM}~\cite{li2017neural} is a sequential recommender model based on the attention mechanism.
\textbf{STAMP}~\cite{liu2018stamp} is a sequential recommender model which combines the modeling of user short- and long-term preferences.
\textbf{SASRec}~\cite{SASRec} is a sequential recommender model based on the self-attention mechanism, which can better capture user long-term preference dependencies.
For fair comparisons, we also include two baselines which incorporate side information.
In specific, 
{\textbf{VTRNN~\cite{cui2016visual}} is a sequential recommender model, where the embeddings of the side information is fused into the model inputs.}
{\textbf{MMSARec}~\cite{han2020sequential} is a sequential recommender model by encoding the side information into the model architectures.}

For our model, we use Causer (LSTM) and Causer (GRU) to discriminate different implementations of $g$.
In sequential models, if there are multiple items at a step, we firstly organize them into a multi-hot vector, and then multiply it with a parameter matrix to derive the input embedding, which is similar to our model.
{For the baselines of VTRNN and MMSARec, the side information is set as the same item raw features as used in our models.}

\noindent
\textbf{Implementation details.}
In the experiment, we firstly organize the interactions of each user according to the time information.
If many items are interacted by a user at the same time, then they are assembled into a multi-hot vector before inputting into the model.
If there is only one item at some time, then the input is a one-hot vector.
Following the common practice~\cite{li2017neural,liu2018stamp}, the last and second last interaction sets of each user are used for model testing and validation, while the others are left for training. A slight difference between our setting and previous work is that the predicted results may be compared with an item set, instead of just a single item.
The widely used metrics including $F_1$ and NDCG are leveraged for model evaluation.
More specifically, suppose $A_{u}$ and $B_{u}$ are the set of items recommended to user $u$ and the ones actually purchased by them in the testing set.
$Z$ is the number of recommended items.
$R(i)$ is the relevance score, where $R(i) = 1$ if the $i$th recommended item belong to $B_{u}$, otherwise $R(i)=0$.
Then the formulas for computing $F_1$ and NDCG are:
\begin{eqnarray}
\begin{aligned}
\text{P(u)}@{Z} &= \frac{\left| A_{u}\cap B_{u} \right|}{\left| A_{u} \right|}\\
\text{R(u)}@{Z} &= \frac{\left| A_{u}\cap B_{u} \right|}{\left| B_{u} \right|}\\
\mathrm {F_1@Z} &= \frac{1}{|\mathcal{U}|}\sum_{u \in \mathcal{U}}\frac{2\cdot \text{P(u)}@{Z}\cdot \text{R(u)}@{Z}} {\text{P(u)}@{Z} + \text{R(u)}@{Z}} \\
\text{DCG}_u@{Z} & = \sum_{i=1}^{\text{Z}}{\frac{R(i)}{log_2(i+1)}} \\
\text{NDCG}@{Z}  & = \frac{1}{|\mathcal{U}|} \sum_{u\in \mathcal{U}}\frac{\text{DCG}_u@Z} {\text{IDCG}}
\nonumber
\end{aligned}
\end{eqnarray}
where IDCG is the max value of $\text{DCG}_u@{Z}$~\cite{buckland1994relationship}.
In the experiments, five items are recommended from each model to compare with the ground truth, that is $Z=5$.
The parameters in our model are determined based on grid search, and we summarize the parameter tuning ranges in Table~\ref{tuning}.
The parameters in the baselines are set as their default values in the original papers or tuned in the same ranges as our model's.

\begin{table}[t]
	\centering
	\caption{{The tuning ranges of our model parameters.}}
	\vspace{-0.cm}
	\renewcommand\arraystretch{1.3}
	\scalebox{1.}{
		\begin{threeparttable} 
			\begin{tabular}{p{2.cm}<{\centering}|p{5.8cm}<{\centering}}
				\hline\hline
                Parameter & Tuning range\\ \hline
                Batch size &$\{32,64,128,256,512,1024\}$\\ 
                Learning rate &$\{10^{-1},10^{-2},10^{-3},10^{-4},10^{-5}\}$\\ 
                Embedding size&$\{32,64,128,256\}$\\ 
                $\epsilon$ &$\{0.1,0.2,0.3,0.4,0.5,0.6,0.7,0.8,0.9\}$\\ 
                $\eta$&$\{10^{-8},10^{-6},10^{-4},10^{-2},1,10^{2},10^{4},10^{6},10^{8}\}$\\ 
                K & $\{2,3,4,5,6,7,8,9,10,20,30,...,100\}$ \\ 
                $\lambda $ &$\{10^{-8},10^{-6},10^{-4},10^{-2},1,10^{2},10^{4},10^{6},10^{8}\}$\\ \hline
				\hline
			\end{tabular}
	\end{threeparttable} }
	\label{tuning}
	\vspace{-0.cm}
\end{table}

\begin{table*}[!t]
	\caption{{{Overall comparison between our models and the baselines. 
				We use bold fonts to label the best performance.
				All the numbers are percentage values with ``\%'' omitted. 
				$^*$ indicates the performance improvements of our model against the best baseline are significant under paired-t test with ``$p<0.05$''.}}}
	\center
	\renewcommand\arraystretch{1.2}
	\vspace{-0.cm}
	\setlength{\tabcolsep}{13.1pt}
	\begin{threeparttable}  
		\scalebox{1.}{
			\begin{tabular}
				{p{0.2cm}<{\centering}p{0.4cm}<{\centering}|
					p{0.2cm}<{\centering}p{0.8cm}<{\centering}|
					p{0.2cm}<{\centering}p{0.8cm}<{\centering}|
					p{0.2cm}<{\centering}p{0.8cm}<{\centering}|
					p{0.2cm}<{\centering}p{0.8cm}<{\centering}|
					p{0.2cm}<{\centering}p{0.8cm}<{\centering}
				} \hline\hline
			
				\multicolumn{2}{c|}{Datasets}&\multicolumn{2}{c|}{Epinions}&\multicolumn{2}{c|}{Baby}&
				\multicolumn{2}{c|}{Patio}&\multicolumn{2}{c|}{Video}&\multicolumn{2}{c}{Foursquare}\\ \hline
				\multicolumn{2}{c|}{Metric (@5)} & F1 & NDCG & F1 & NDCG & F1 & NDCG & F1 & NDCG& F1 & NDCG \\ \hline
				\multicolumn{2}{c|}{BPR}            & 0.63 & 1.28 & 0.72 & 1.33 & 0.37  & 0.61 & 1.08 & 2.11& 2.45 &4.76 \\ 
				\multicolumn{2}{c|}{NCF}           & 1.00 & 1.42 & 0.90 & 1.67 & 0.53 & 1.09 & 0.92 & 1.97 & 3.05 &6.28\\ 
				\multicolumn{2}{c|}{GRU4Rec}  & 0.97 & 1.61 & 0.90  & 1.68 & 0.37 & 0.75 & 0.95 & 2.01 & 3.05 &6.32 \\
				\multicolumn{2}{c|}{STAMP}      & 1.05 & 1.95& 0.88 & 1.67  & 0.47 & 1.03 & 0.95 & 1.99 & \textbf{3.08}  &6.32 \\ 
				\multicolumn{2}{c|}{SASRec}     & 1.00 & 1.45& 0.90 & 1.67 & 0.48 & 0.89 & 1.02 & 2.02& 3.05  &6.26  \\ 
				\multicolumn{2}{c|}{NARM}        & 1.08 & 1.93  & 0.90 & 1.68 & 0.38  & 0.72 & 1.48  & 2.90 & 2.80 &6.06  \\ 
				\multicolumn{2}{c|}{VTRNN}      & 0.55 & 1.52 & 0.83 & 1.51  & 0.60 & 1.05 & 1.53 & 2.91& 3.05  &5.26 \\ 
				\multicolumn{2}{c|}{MMSARec} & 0.97 & 1.48 & 0.90 & 1.66 & 0.42 & 0.69 & 1.88 & 3.42& 3.05 &6.30\\ 
				\hline
				\multicolumn{2}{c|}{Causer (LSTM)} & \textbf{1.17}$^*$ & 2.00& 0.90 & 1.68  & 0.69 & 1.35 & 1.91 & 3.51 & 3.05 &6.34 \\ 
				\multicolumn{2}{c|}{Causer (GRU)} & 1.13 & \textbf{2.17}$^*$& \textbf{0.92} & \textbf{1.71}$^*$  & \textbf{0.71}$^*$ & \textbf{1.46}$^*$  & \textbf{1.95}$^*$ & \textbf{3.63}$^*$&  \textbf{3.08}  &  \textbf{6.36} \\

				\hline \hline 
			\end{tabular}
		}        
	\end{threeparttable}    
	\label{tab:ab-result}   
	\vspace{-0.cm}
\end{table*}

\subsection{Overall Performance Comparison}\label{overall}
The overall comparison results are presented in Table~\ref{tab:ab-result}, from which we can see:
{the performance of all the models are not high, which is because of the sparse and noisy nature of the recommendation datasets, and verifies the difficulties of the recommendation task itself.}
Neural models like NCF can usually obtain better performance than the shallow model BPR, which agrees with the previous work~\cite{ncf}, and verifies the usefulness of modeling user-item non-linear relationships.
Among sequential recommender models, the best performance is usually achieved when the model is based on the attention mechanism or has side information.
This is as expected, since the attention mechanism can highlight the items which are more important for the next item prediction, and the side information can provide additional signals to profile the items, which facilitates more comprehensive collaborative feature modeling and improved the recommendation performance. 

Encouragingly, our model can achieve the best performance on all the metrics across different datasets, where the improvements are mostly significant.
On average, our model can improve the best baseline by about {6.1\% and 11.3\%} on $F_1$ and NDCG, respectively.
Comparing with the baselines, we introduce a causal discovery module to filter the causally irrelevant history information.
This module lowers the negative effects from the spurious item correlations, and makes our model focus on the real causal relations among user behaviors, which improves the recommendation performance.
Between different implementations of $g$, we find that GRU is superior than LSTM in most cases.
We speculate that GRU is a lighter architecture, which can be more appropriate for the sparse recommendation datasets.
LSTM contains too many parameters, which may easily over-fit the training data, and lead to the lowered performance on the testing set.

\subsection{Parameter Analysis}
In this section, we analyze the influences of the latent cluster number K, the threshold $\epsilon$ and the temperature $\eta$ on the model performance.
When studying one parameter, we fix the other ones as their optimal values.
The results are reported based on NDCG and the datasets of Baby and Epinions.
The conclusions on the other metric and datasets are similar and omitted.

\subsubsection{Influence of the number of latent clusters K}
In our model, the hyper parameter K basically encodes our belief on how may latent clusters is appropriate to cover the item properties.
In order to study its influence, we tune it in [2,3,4,5,6,7,8,9,10,20,30,40,50,60,70,80,90,100]. 
From the results shown in Figure~\ref{fig_k}, we can see:
the best performance is usually achieved when K is {relatively} small.
We speculate that too large K (\emph{e.g.}, 100) may introduce too many redundant parameters, which may enhance the risk of model over-fitting, and lower the model generalization capability on the testing set.
For Baby, the best K is between 4 and 6, while for Epinions, $K\in [15, 20]$ can usually lead to the better performances.
We speculate that, for Baby, the products are all about baby toys, baby feedings and so on.
The items are quite homogeneous, thus only a small number of clusters can be enough to characterize them.
For Epinions, the items are much more diverse, ranging from the fields of electronics and office to the sports and travel.
As a result, more clusters are needed to sufficiently cover the item features.
In addition, we find that K cannot be too small, for example, when $K\in [2, 4]$, the performances on both datasets are not satisfied.
The reason can be that when K is too small, the clusters are not expressive enough to characterize the item space.
Many semantically irrelevant items are mixed into the same cluster, which brings difficulties for learning causal relations among them, and leads to less effective history information purification and lowered performances.

\subsubsection{Influence of the threshold $\epsilon$}
According to equation~(\ref{cx}), the threshold $\epsilon$ determines how severe we filter the causally irrelevant history information.
When we set $\epsilon$ as a larger value, only the items which are more likely to be the cause of the target item are remained.
However, at the same time, the number of left items can be small, which may impact the accuracy of the item correlation modeling.
When we set $\epsilon$ as a small value, more items are incorporated into the training process, but they may be causally less relevant, which may bring more noises.
In order to study the influence of $\epsilon$, we tune it in the range of [0.1,0.2,0.3,0.4,0.5,0.6,0.7,0.8,0.9].
We present the results in Figure~\ref{fig_epsilon}, where we can see:
for both datasets and different sequential architectures, our model can achieve the best performance when $\epsilon$ is moderated.
This agrees with our above analysis.
By trading-off the number of training samples and the causal degree of their relations, $\epsilon$ provides us with the opportunity to achieve an equilibrium point, which {leads} to better recommendation performances.

\subsubsection{Influence of the temperature parameter $\eta$}
In our model, $\eta$ determines the softness of the cluster assignment distribution.
In this section, we study its influence on the model performance by tuning it in the range of $[10^{-8},10^{-6},10^{-4},10^{-2},1,10^{2},10^{4},10^{6},10^{8}]$.
From the results shown in Figure~\ref{fig_eta}, we can see:
for both datasets, when $\eta$ is small, the performance continually rises up as we increase $\eta$.
After reaching the optimal point, the model performance goes down as $\eta$ becomes larger.
For the same dataset, when we implement $g$ with different architectures, the best performances are achieved with similar $\eta$'s, which may suggest that $\eta$ is a robust parameter w.r.t. the sequential architecture.
For different datasets, the optimal $\eta$ varies a lot.
For Baby, the best performance is achieved when $\eta\in [1,10^{2}]$, while for Epinions, $\eta\in [10^{2}, 10^{4}]$ can usually lead to the better performances.
This observation manifest that $\eta$ is very sensitive to the datasets, and one has to carefully tune it when applying our model in different scenarios.

\begin{figure}[t]
	\centering
	\setlength{\fboxrule}{0.pt}
	\setlength{\fboxsep}{0.pt}
	\fbox{
		\includegraphics[width=.95\linewidth]{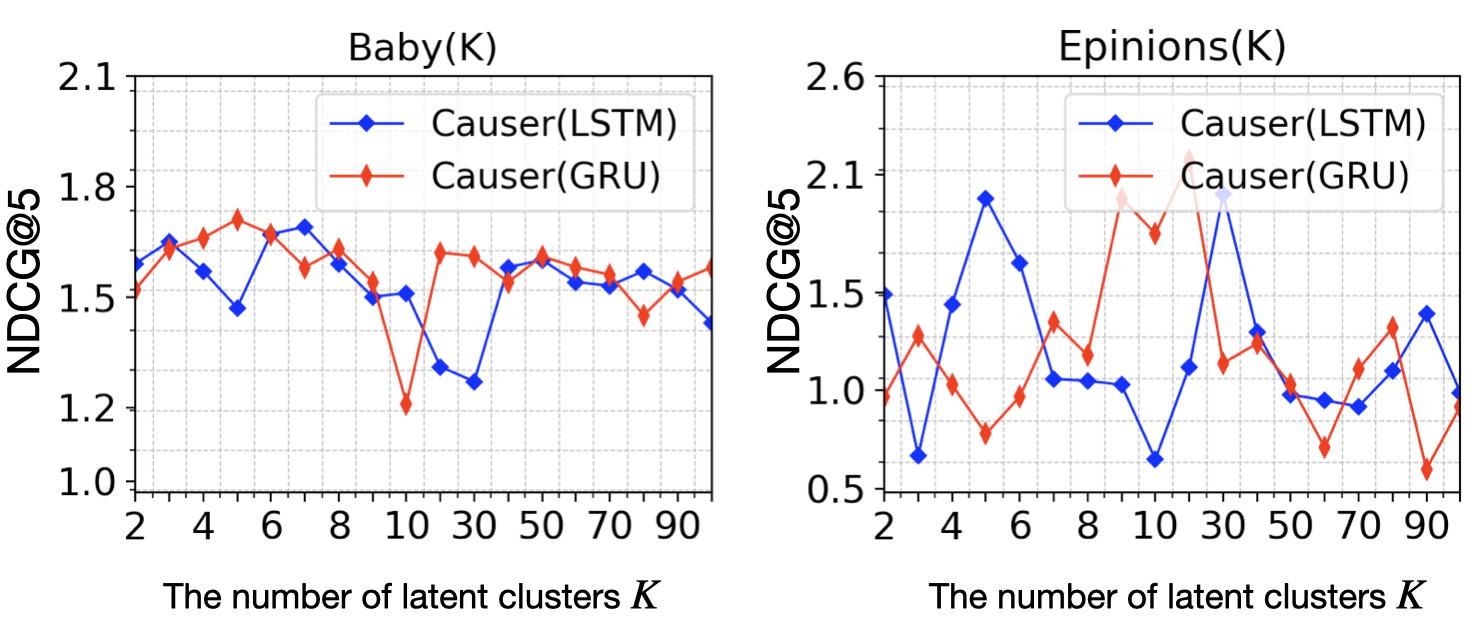}
	}
	\vspace{-0.1cm}
	\caption{{Influence of the number of latent clusters $K$ on the model performance in terms of NDCG@5.
	The performances of different sequential architecture implementations are labeled with different colors.
	{The results are percentage values with ``\%'' omitted.}
	}}
	\label{fig_k}
	\vspace{-0.1cm}
\end{figure}

\begin{figure}[t]
	\centering
	\setlength{\fboxrule}{0.pt}
	\setlength{\fboxsep}{0.pt}
	\includegraphics[width=.95\linewidth]{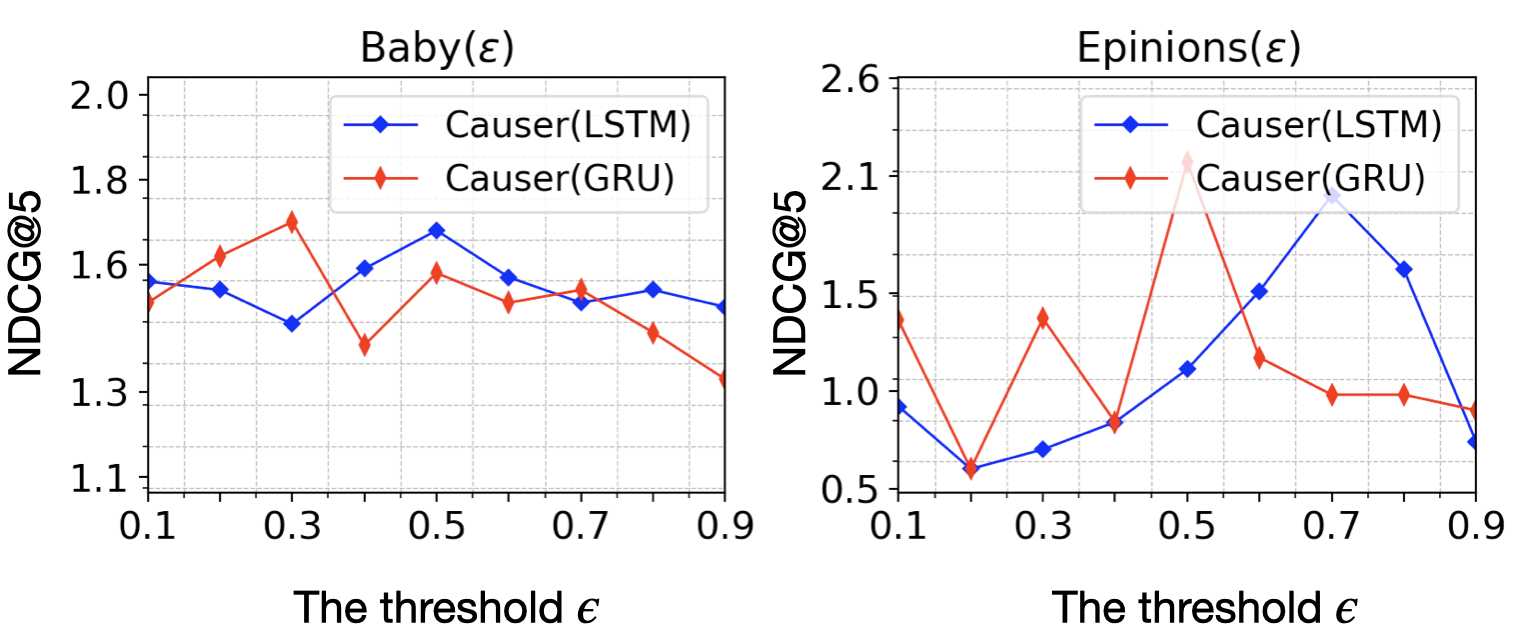}
	\vspace{-0.1cm}
	\caption{{Influence of the threshold $\epsilon$ on the model performance in terms of NDCG@5.
	The performances of different sequential architecture implementations are labeled with different colors.
	{The results are percentage values with ``\%'' omitted.}
	}}
	\label{fig_epsilon}
	\vspace{-0.cm}
\end{figure}

\begin{figure}[t]
	\centering
	\setlength{\fboxrule}{0.pt}
	\setlength{\fboxsep}{0.pt}
	\includegraphics[width=1.\linewidth]{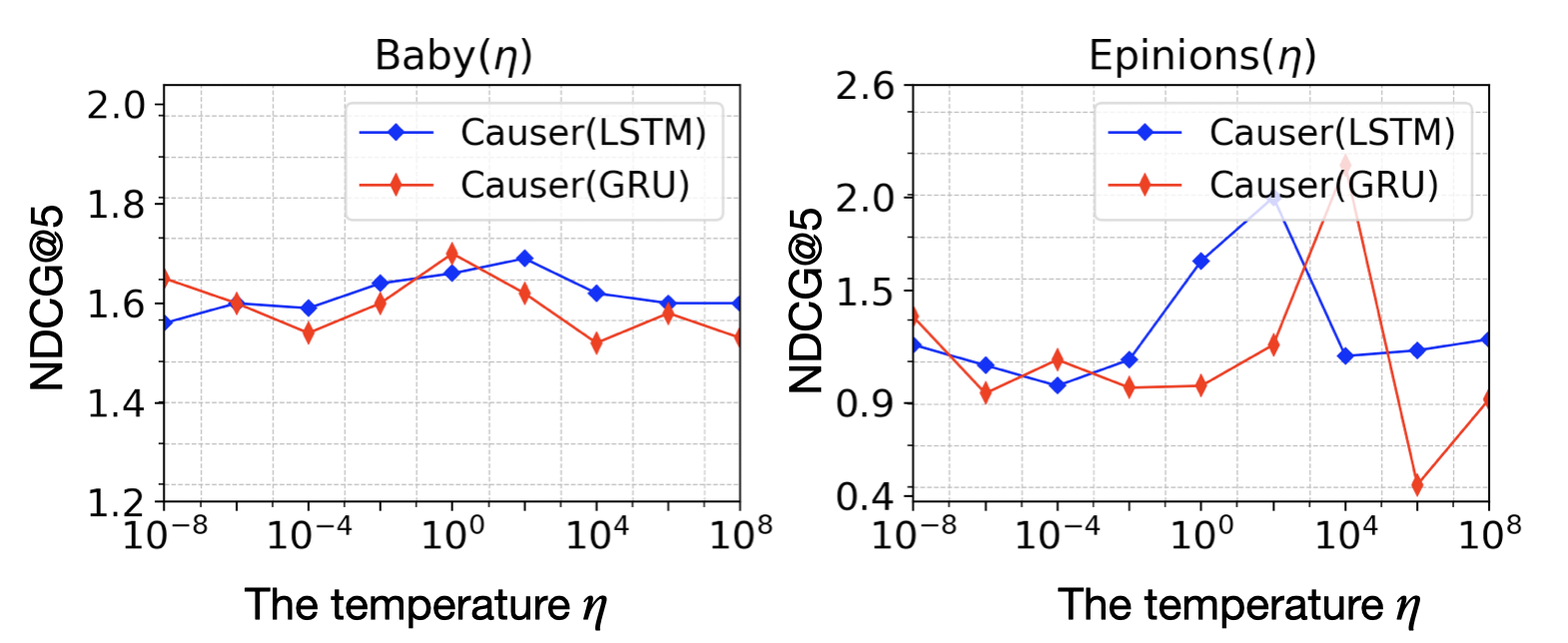}
	\vspace{-0.3cm}
	\caption{{Influence of the temperature $\eta$ on the model performance in terms of NDCG@5.
	The performances of different sequential architecture implementations are labeled with different colors.
	{The results are percentage values with ``\%'' omitted.}}}
	\label{fig_eta}
	\vspace{-0.1cm}
\end{figure}

\begin{table}[!t]
	\caption{Ablation studies: performance comparison between our model and its variants in terms of NDCG@5. {The results are percentage values with ``\%'' omitted.}}
	\center
	\renewcommand\arraystretch{1.2}
	\vspace{-0.cm}
	\setlength{\tabcolsep}{7.1pt}
	\begin{threeparttable}  
		\scalebox{1.}{
			\begin{tabular}
				{p{0.6cm}<{\centering}|
					p{1.cm}<{\centering}|p{1.cm}<{\centering}|
					p{1.cm}<{\centering}|p{1.cm}<{\centering}
				} \hline\hline
				\multicolumn{1}{c|}{$g$}&\multicolumn{2}{c|}{LSTM}&\multicolumn{2}{c}{GRU}\\ \hline
				\multicolumn{1}{c|}{Datasets} & Baby & Epinions & Baby & Epinions  \\ \hline
				\multicolumn{1}{c|}{Causal (-rec)} & 1.56 & 1.23 & 1.60 & 1.36  \\ 
				\multicolumn{1}{c|}{Causal (-clus)} & 1.59 & 1.47 & 1.64 & 1.35  \\ 
				\multicolumn{1}{c|}{Causal (-att)} & 1.65 & 1.89 & 1.69 & 1.95  \\ 
				\multicolumn{1}{c|}{Causal (-causal)} & 1.65 & 1.52 & 1.67 & 1.61  \\ \hline
				\multicolumn{1}{c|}{Causal} & \textbf{1.68} & \textbf{2.00} & \textbf{1.71} & \textbf{2.17}  \\ \hline\hline
			\end{tabular}
		}           
	\end{threeparttable}  
	\vspace{-0.cm}
	\label{ab-result-cx}
\end{table}

\subsection{Ablation Studies}
After evaluating our model as a whole, we would like to study whether different designs in our model are necessary.
To this end, we compare our model with its four {variants}:
\underline{\textit{Causer (-clus)}} is a method by removing the clustering loss~(\ref{deepclu}).
\underline{\textit{Causer (-rec)}} is a model, where we remove the reconstruction loss~(\ref{dec}).
\underline{\textit{Causer (-att)}} is a variant, where we do not use the attention mechanism. 
\underline{\textit{Causer (-causal)}} is a method, where we drop the causal relation $\hat{{W}}_{\vec{\bm{v}}_k^{t}b}$ in equation~(\ref{cx}).
Similar to the above experiment, we report the results based on NDCG and the datasets of Baby and Epinions, and the model parameters follow the settings in section~\ref{overall}.

From the results shown in Table~\ref{ab-result-cx}, we can see:
Causer performs better than Causer (-clus), which verifies the effectiveness of the clustering loss.
We speculate that the clustering loss ensures that similar item embeddings can be pulled into the same cluster, which is critical for accurately estimating the causal relations, and deriving clean and causal history representations to improve the recommendation performance.
Causer (-rec) performs worse than Causer.
This maybe because by introducing the reconstruction loss, the item embeddings are forced to encode the basic item properties, which is the foundation for accurately clustering items and learning causal relations, and thus important for the final results. 
It is interesting to see that when we remove the attention mechanism, the performance of Causer (-att) is lowered.
We speculate that for different local contexts, the item importances may also vary.
For the example in Figure 1, if there is another printer with different brand, then the importance of the original printer should be lowered, since the ink box may also be caused by the new printer.
This intuitive example suggests that it is reasonable to refine the global causal relations with the item local importance, which may bring improved recommendation performance.
As expected, the causal relation is indeed helpful, which is evidenced by the lowered performance of Causer (-causal) as compared with Causer. 
This observation demonstrates the effectiveness of our idea on capturing item causal relations.
By combining the clustering loss, reconstruction loss, local attentions and global causal relations, our final model can achieve the best performance on different datasets and sequential architectures. 
The above results confirm the effectiveness of our model designs, and demonstrate that all of them can contribute the final performance.

\subsection{Explanation Evaluation}
In the above sections, we have demonstrated the effectiveness of our model in boosting the recommendation performance, and also studied the influences of the hyper parameters.
In addition to performance improvement, our model can also provide more accessible recommendation explanations.
In order to evaluate the explanations generated by our model, we conduct both quantitative and qualitative experiments.
More specifically, we based our experiment on the dataset of Baby.
The model parameters are set as their optimal values tuned in section~\ref{overall}.

\subsubsection{Quantitative Analysis}
Existing recommendation datasets are mostly designed for evaluating the performance.
While we have noticed a few recommendation explanation datasets\footnote{https://competition.huaweicloud.com/information/1000041488/introduction}~\cite{li2021extra}, they are not for sequential recommendation, and do not contain ground truth on item causal relations.

To solve this problem, we manually label a new dataset for our quantitative analysis.
In specific, we select 1000 samples from the testing set of Baby.
For easy labeling and evaluation, we select the samples, where at each step, there is only one interacted item.
In order to accurately label causal relations between different items, we firstly teach the workers to separate causal and non-causal relations by presenting a large amount of commonly recognized examples. 
Then, for a testing sample, we ask the workers to label out 3 items from the history information, which are most likely to be the real cause of the target item. For each sample, we randomly assign three workers, and only the commonly labeled items are remained for controlling the quality of the dataset.
At last, we obtain a dataset containing {793} samples, and for each sample, we have on average {1.8} causal items in the history information.

Based on the above dataset, our quantitative analysis is conducted by comparing Causer (-att), Causer (-causal) and Causer.
In the experiments, for each item in the history information, we firstly compute the explanation score $\hat{{W}}_{\vec{\bm{v}}_k^{t}b}$, $\alpha_t$ and $\hat{{W}}_{\vec{\bm{v}}_k^{t}b}\alpha_t$ for Causer (-att), Causer (-causal) and Causer, respectively, and then the items with the largest scores are used to explain the recommendations.
For each model, we select 3 items to compare with the ground truth, where $F_1$ and NDCG are leveraged as the evaluation metrics.
The comparison results are presented in Figure~\ref{explane}, from which we can see:
Causer (-att) performs better than Causer (-causal) on both evaluation metrics, which suggests that our designed causal discovery module can be more important than the attention mechanism in providing causal item explanations.
The complete Causer model can consistently provide better explanations than its variants.
For Causer (-causal), the relations captured by the attention mechanism are based more on the item-concurrence, which cannot grantee causalities.
In our model, we intentionally design a causal discovery module to learn item causal relations, which can be more aligned with the dataset labeled based on causalities.
Besides, the lowered performance of Causer (-att) manifests that the local attention mechanism is also useful.
The best explanations are usually achieved by combing the local attention and global causal mechanisms.

\begin{figure}[t]
	\centering
	\setlength{\fboxrule}{0.pt}
	\setlength{\fboxsep}{0.pt}
	\includegraphics[width=1.\linewidth]{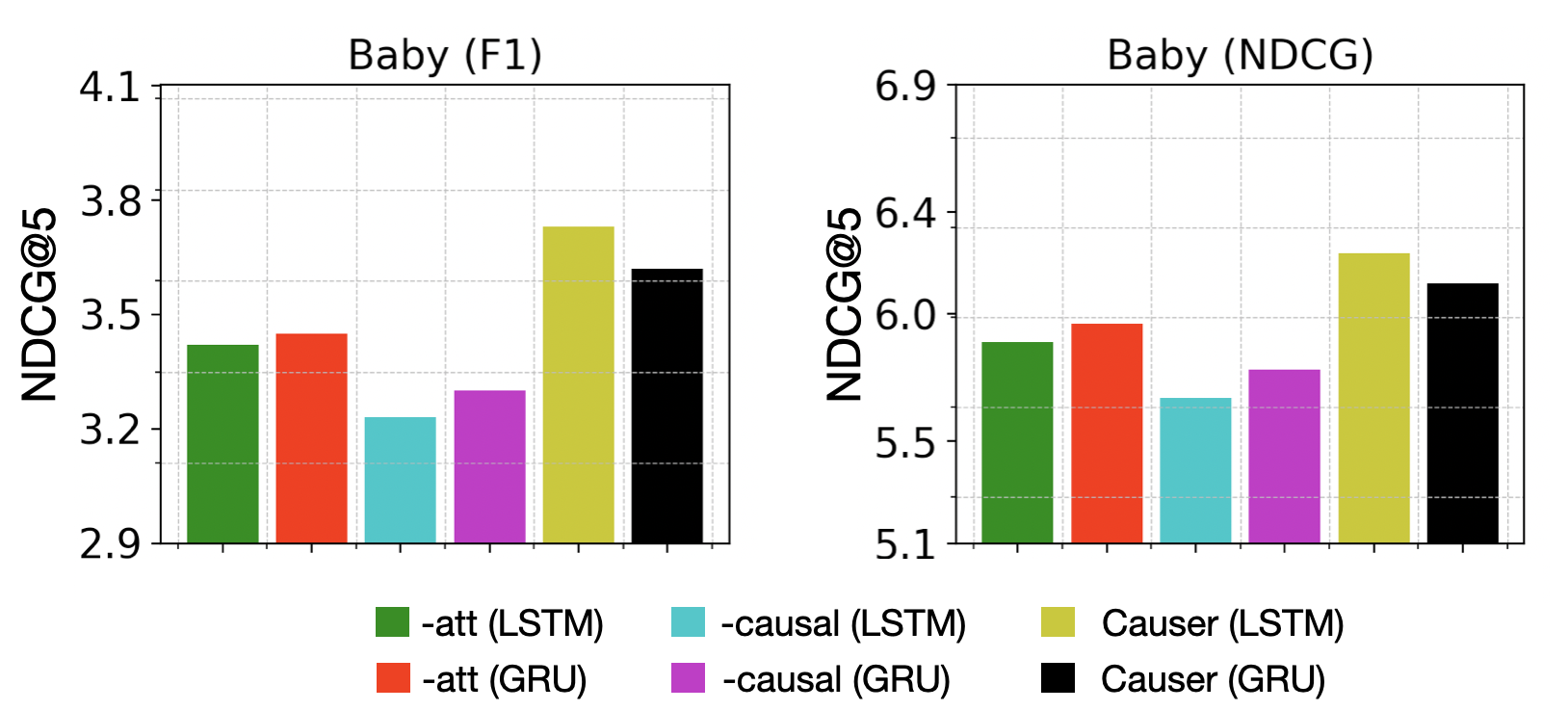}
	\vspace{-0.2cm}
	\caption{Quantitative evaluation results on the recommendation explanations.
    Different variants and sequential architectures are labeled with different colors.
    {The results are percentage values with ``\%'' omitted.}
	}
	\label{explane}
	\vspace{-0.1cm}
\end{figure}

\begin{figure}[t]
	\centering
	\setlength{\fboxrule}{0.pt}
	\setlength{\fboxsep}{0.pt}
	\includegraphics[width=1.\linewidth]{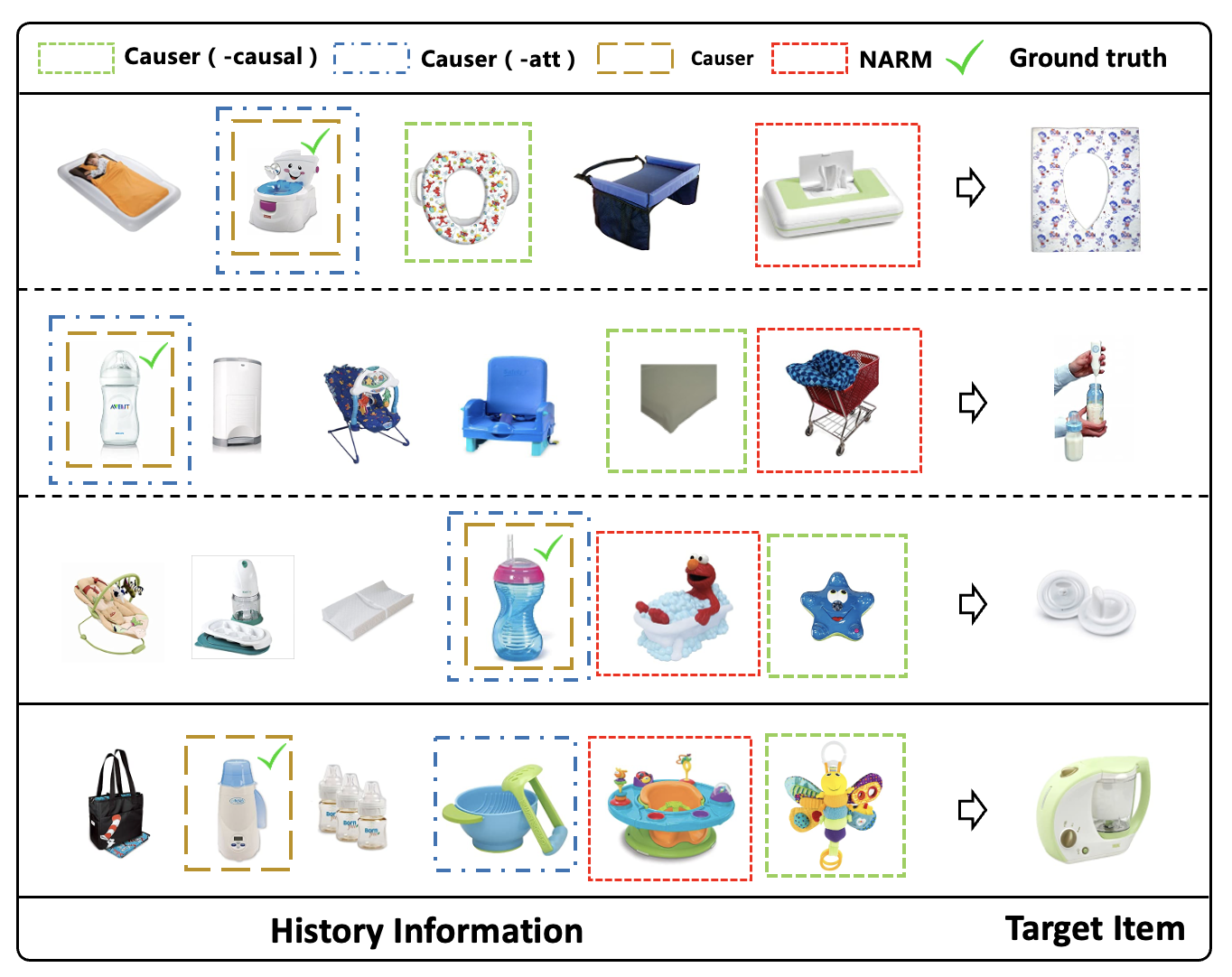}
	\vspace{-0.4cm}
	\caption{Qualitative evaluation results on the recommendation explanations. 
	Four cases are presented to illustrate the advantages of our model comparing with the purely attentive- and causal-based methods.}
	\label{case}
	\vspace{-0.cm}
\end{figure}

\subsubsection{Qualitative Analysis}
In order to provide more intuitive understandings on the provided explanations, in this section, we conduct many case studies, where the parameter settings follow the above experiments, and we compare the explanations generated by NARM, Causer (-att), Causer (-causal) and Causer.
From the results shown in Figure~\ref{case}, we can see:
in the first example, the target item is a toilet seat.
With the help of the causal discovery module, our model and Causer (-att) can successfully find out the real cause item, that is, the baby toilet.
{By removing the causal relation from our model, Causer (-causal) labels another toilet seat.
While it is related to the target item, the relation is not causal, and it is unreasonable to explain a toilet seat with another one.
NARM leverages the toilet paper as the explanation, which is also inappropriate.
}
Similar results can also be found in the second and third examples, where our model and Causer (-att) can accurately discover the causal relations from the feeding bottle to the milking machine or baby nipple.
In the last example, for the target item, \emph{i.e.}, a baby bottle warmer, our model explains it with the baby bottle, while Causer (-att) regards the bowl as the explanation, which is less reasonable.
This observation may suggest that both the attention and causal discovery module are useful for the accurate recommendation explanations.

\section{Conclusion, Limitation and Future Work}\label{conclu}
\textbf{Conclusion}. In this paper, we propose to improve sequential recommendation with causal discovery to capture causalities among user sequential behaviors.
We realize our idea by designing cluster-level causal graph and seamlessly infusing it into ordinary sequential recommender models.
Theoretical analysis {is} provided to demonstrate the identifiability of the causal graph learned in our framework.
Extensive experiments demonstrate that our model can improve the recommendation performance and explainability.

\textbf{Limitation and Future Work}. 
{This paper actually makes a first step towards causality enhanced sequential recommendation.
However, there are still many limitations, which left much room for improvement.
To begin with, we assume that the causal graph is static, which can not model the dynamic user preference.
In the future, an interesting direction is to introduce dynamic causal graph into our model, where the causal relation can be altered when the interaction times are different.
In our model, we do not consider the interactions which are not recorded in the datasets, if the cause of an item is not recorded, then the causal discovery module may fail. In the future, we plan to explicitly model the latent confounders of the causal graph, which may lead better recommendation performance and explainability.
At last, our model can only capture one-to-one causal relations, in the future, a promising direction is to extend our model to capture multi-to-one, one-to-multi and multi-to-multi causal relations.}

\section*{Acknowledgment}
This work is supported in part by National Natural Science Foundation of China (No. 62102420), Beijing Outstanding Young Scientist Program NO. BJJWZYJH012019100020098, Intelligent Social Governance Platform, Major Innovation \& Planning Interdisciplinary Platform for the "Double-First Class" Initiative, Renmin University of China, and Public Computing Cloud, Renmin University of China. 
The work is sponsored by Huawei Innovation Research Programs. We appreciate the support from Mindspore\footnote{\url{https://www.mindspore.cn}}, which is a new deep learning computing framework.

\bibliographystyle{IEEEtran}
\balance
\bibliography{acmart}

\begin{thebibliography}{10}
\providecommand{\url}[1]{#1}
\csname url@samestyle\endcsname
\providecommand{\newblock}{\relax}
\providecommand{\bibinfo}[2]{#2}
\providecommand{\BIBentrySTDinterwordspacing}{\spaceskip=0pt\relax}
\providecommand{\BIBentryALTinterwordstretchfactor}{4}
\providecommand{\BIBentryALTinterwordspacing}{\spaceskip=\fontdimen2\font plus
\BIBentryALTinterwordstretchfactor\fontdimen3\font minus
  \fontdimen4\font\relax}
\providecommand{\BIBforeignlanguage}[2]{{%
\expandafter\ifx\csname l@#1\endcsname\relax
\typeout{** WARNING: IEEEtran.bst: No hyphenation pattern has been}%
\typeout{** loaded for the language `#1'. Using the pattern for}%
\typeout{** the default language instead.}%
\else
\language=\csname l@#1\endcsname
\fi
#2}}
\providecommand{\BIBdecl}{\relax}
\BIBdecl

\bibitem{rendle2010factorizing}
S.~Rendle, C.~Freudenthaler, and L.~Schmidt-Thieme, ``Factorizing personalized
  markov chains for next-basket recommendation,'' in \emph{Proceedings of the
  19th international conference on World wide web}, 2010, pp. 811--820.

\bibitem{hidasi2015session}
B.~Hidasi, A.~Karatzoglou, L.~Baltrunas, and D.~Tikk, ``Session-based
  recommendations with recurrent neural networks,'' \emph{arXiv preprint
  arXiv:1511.06939}, 2015.

\bibitem{liu2018stamp}
Q.~Liu, Y.~Zeng, R.~Mokhosi, and H.~Zhang, ``Stamp: short-term attention/memory
  priority model for session-based recommendation,'' in \emph{Proceedings of
  the 24th ACM SIGKDD International Conference on Knowledge Discovery \& Data
  Mining}, 2018, pp. 1831--1839.

\bibitem{rychlak1994logical}
J.~F. Rychlak, \emph{Logical learning theory: A human teleology and its
  empirical support}.\hskip 1em plus 0.5em minus 0.4em\relax U of Nebraska
  Press, 1994.

\bibitem{rychlak1986logical}
------, ``Logical learning theory: A teleological alternative in the field of
  personality,'' \emph{Journal of personality}, vol.~54, no.~4, pp. 734--762,
  1986.

\bibitem{li2017neural}
J.~Li, P.~Ren, Z.~Chen, Z.~Ren, T.~Lian, and J.~Ma, ``Neural attentive
  session-based recommendation,'' in \emph{Proceedings of the 2017 ACM on
  Conference on Information and Knowledge Management}, 2017, pp. 1419--1428.

\bibitem{zheng2018dags}
X.~Zheng, B.~Aragam, P.~Ravikumar, and E.~P. Xing, ``Dags with no tears:
  Continuous optimization for structure learning,'' \emph{arXiv preprint
  arXiv:1803.01422}, 2018.

\bibitem{ng2019graph}
I.~Ng, S.~Zhu, Z.~Chen, and Z.~Fang, ``A graph autoencoder approach to causal
  structure learning,'' \emph{arXiv preprint arXiv:1911.07420}, 2019.

\bibitem{brouillard2020differentiable}
P.~Brouillard, S.~Lachapelle, A.~Lacoste, S.~Lacoste-Julien, and A.~Drouin,
  ``Differentiable causal discovery from interventional data,'' \emph{arXiv
  preprint arXiv:2007.01754}, 2020.

\bibitem{lee2013local}
J.~Lee, S.~Kim, G.~Lebanon, and Y.~Singer, ``Local low-rank matrix
  approximation,'' in \emph{International conference on machine
  learning}.\hskip 1em plus 0.5em minus 0.4em\relax PMLR, 2013, pp. 82--90.

\bibitem{haeffele2014structured}
B.~Haeffele, E.~Young, and R.~Vidal, ``Structured low-rank matrix
  factorization: Optimality, algorithm, and applications to image processing,''
  in \emph{International conference on machine learning}.\hskip 1em plus 0.5em
  minus 0.4em\relax PMLR, 2014, pp. 2007--2015.

\bibitem{yang2016joint}
J.~Yang, D.~Parikh, and D.~Batra, ``Joint unsupervised learning of deep
  representations and image clusters,'' in \emph{Proceedings of the IEEE
  conference on computer vision and pattern recognition}, 2016, pp. 5147--5156.

\bibitem{yang2017towards}
B.~Yang, X.~Fu, N.~D. Sidiropoulos, and M.~Hong, ``Towards k-means-friendly
  spaces: Simultaneous deep learning and clustering,'' in \emph{international
  conference on machine learning}.\hskip 1em plus 0.5em minus 0.4em\relax PMLR,
  2017, pp. 3861--3870.

\bibitem{hochreiter1997long}
S.~Hochreiter and J.~Schmidhuber, ``Long short-term memory,'' \emph{Neural
  computation}, vol.~9, no.~8, pp. 1735--1780, 1997.

\bibitem{chung2014empirical}
J.~Chung, C.~Gulcehre, K.~Cho, and Y.~Bengio, ``Empirical evaluation of gated
  recurrent neural networks on sequence modeling,'' \emph{arXiv preprint
  arXiv:1412.3555}, 2014.

\bibitem{quadrana2017personalizing}
M.~Quadrana, A.~Karatzoglou, B.~Hidasi, and P.~Cremonesi, ``Personalizing
  session-based recommendations with hierarchical recurrent neural networks,''
  in \emph{Proceedings of the Eleventh ACM Conference on Recommender Systems},
  2017, pp. 130--137.

\bibitem{sun2019bert4rec}
F.~Sun, J.~Liu, J.~Wu, C.~Pei, X.~Lin, W.~Ou, and P.~Jiang, ``Bert4rec:
  Sequential recommendation with bidirectional encoder representations from
  transformer,'' in \emph{Proceedings of the 28th ACM international conference
  on information and knowledge management}, 2019, pp. 1441--1450.

\bibitem{spirtes2000causation}
P.~Spirtes, C.~N. Glymour, R.~Scheines, and D.~Heckerman, \emph{Causation,
  prediction, and search}.\hskip 1em plus 0.5em minus 0.4em\relax MIT press,
  2000.

\bibitem{meek2013causal}
C.~Meek, ``Causal inference and causal explanation with background knowledge,''
  \emph{arXiv preprint arXiv:1302.4972}, 2013.

\bibitem{zhang2008completeness}
J.~Zhang, ``On the completeness of orientation rules for causal discovery in
  the presence of latent confounders and selection bias,'' \emph{Artificial
  Intelligence}, vol. 172, no. 16-17, pp. 1873--1896, 2008.

\bibitem{chickering2002optimal}
D.~M. Chickering, ``Optimal structure identification with greedy search,''
  \emph{Journal of machine learning research}, vol.~3, no. Nov, pp. 507--554,
  2002.

\bibitem{chickering1997efficient}
D.~M. Chickering and D.~Heckerman, ``Efficient approximations for the marginal
  likelihood of bayesian networks with hidden variables,'' \emph{Machine
  learning}, vol.~29, no.~2, pp. 181--212, 1997.

\bibitem{heckerman1995learning}
D.~Heckerman, D.~Geiger, and D.~M. Chickering, ``Learning bayesian networks:
  The combination of knowledge and statistical data,'' \emph{Machine learning},
  vol.~20, no.~3, pp. 197--243, 1995.

\bibitem{bouckaert1993probabilistic}
R.~R. Bouckaert, ``Probabilistic network construction using the minimum
  description length principle,'' in \emph{European conference on symbolic and
  quantitative approaches to reasoning and uncertainty}.\hskip 1em plus 0.5em
  minus 0.4em\relax Springer, 1993, pp. 41--48.

\bibitem{BPR}
S.~{Rendle}, C.~{Freudenthaler}, Z.~{Gantner}, and L.~{Schmidt-Thieme}, ``{BPR:
  Bayesian Personalized Ranking from Implicit Feedback},'' \emph{arXiv
  e-prints}, p. arXiv:1205.2618, May 2012.

\bibitem{ncf}
\BIBentryALTinterwordspacing
X.~He, L.~Liao, H.~Zhang, L.~Nie, X.~Hu, and T.-S. Chua, ``Neural collaborative
  filtering,'' in \emph{Proceedings of the 26th International Conference on
  World Wide Web}, ser. WWW '17.\hskip 1em plus 0.5em minus 0.4em\relax
  Republic and Canton of Geneva, CHE: International World Wide Web Conferences
  Steering Committee, 2017, p. 173–182. [Online]. Available:
  \url{https://doi.org/10.1145/3038912.3052569}
\BIBentrySTDinterwordspacing

\bibitem{SASRec}
W.-C. Kang and J.~McAuley, ``Self-attentive sequential recommendation,'' in
  \emph{2018 IEEE International Conference on Data Mining (ICDM)}, 2018.

\bibitem{cui2016visual}
Q.~Cui, S.~Wu, Q.~Liu, and L.~Wang, ``A visual and textual recurrent neural
  network for sequential prediction,'' \emph{arXiv preprint arXiv:1611.06668},
  2016.

\bibitem{han2020sequential}
T.~Han, Y.~Tian, J.~Zhang, and S.~Niu, ``Sequential recommendation with a
  pre-trained module learning multi-modal information,'' in \emph{2020
  International Conferences on Internet of Things (iThings) and IEEE Green
  Computing and Communications (GreenCom) and IEEE Cyber, Physical and Social
  Computing (CPSCom) and IEEE Smart Data (SmartData) and IEEE Congress on
  Cybermatics (Cybermatics)}.\hskip 1em plus 0.5em minus 0.4em\relax IEEE,
  2020, pp. 611--616.

\bibitem{buckland1994relationship}
M.~Buckland and F.~Gey, ``The relationship between recall and precision,''
  \emph{Journal of the American society for information science}, vol.~45,
  no.~1, pp. 12--19, 1994.

\bibitem{li2021extra}
L.~Li, Y.~Zhang, and L.~Chen, ``Extra: Explanation ranking datasets for
  explainable recommendation,'' \emph{arXiv preprint arXiv:2102.10315}, 2021.

\end{thebibliography}

\end{document}